\documentclass[11pt,a4paper]{article}
\usepackage[utf8]{inputenc}
\usepackage[T1]{fontenc}
\usepackage{amsmath,amssymb,amsthm}
\usepackage[margin=2.5cm]{geometry}
\usepackage{xcolor}
\usepackage{enumitem}
\usepackage[colorlinks=true,allcolors=black]{hyperref}

\newtheorem{theorem}{Theorem}[section]
\newtheorem{proposition}[theorem]{Proposition}

\newtheorem{lemma}[theorem]{Lemma}
\theoremstyle{definition}
\newtheorem{definition}[theorem]{Definition}
\theoremstyle{remark}

\newcommand{\J}{\mathbf J}
\newcommand{\rr}{\mathbf r}
\newcommand{\pp}{\boldsymbol\pi}
\newcommand{\Tr}{\operatorname{tr}}
\newcommand{\Fi}{\Phi}
\newcommand{\Cc}{\mathcal C}
\newcommand{\Aa}{\mathcal A}
\newcommand{\Alg}{\mathfrak A}
\newcommand{\Pc}{\mathcal P}
\newcommand{\so}{\mathfrak{so}}
\newcommand{\su}{\mathfrak{su}}

\newcommand{\sixj}[6]{\begin{Bmatrix}#1&#2&#3\\#4&#5&#6\end{Bmatrix}}

\title{Symmetry algebra of a charged harmonic oscillator\\ in the field of the Dirac magnetic monopole}
\author{Ian Marquette\thanks{Department of Mathematical and Physical Sciences, La Trobe University, Bendigo, VIC 3552, Australia.\texttt{i.marquette@latrobe.edu.au}}\and Luc Vinet\thanks{Centre de recherches math\'ematiques, Universit\'e de Montr\'eal, Montr\'eal (Qu\'ebec) H3C 3J7, Canada. \texttt{luc.vinet@umontreal.ca}}}
\date{}

\begin{document}
\maketitle

\begin{abstract}

A charged particle bound by an isotropic harmonic potential to a Dirac magnetic monopole, with the centrifugal term that makes the orbits close, has degeneracies beyond rotational invariance: for integer $q$, the product of the electric and magnetic charges, the levels are those of the oscillator with the states of angular momentum $l<|q|$ removed. The conserved rank-2 tensor $T$ built from the ladder operators that accounts for them was identified thirty-five years ago, but the algebra it generates with the angular momentum $\mathbf{J}$ was not obtained. We obtain its relations in covariant form and identify it. It is polynomial but not quadratic, and the shift $T'=T+2(H-1)J^{(2)}$ by the quadrupole of $\mathbf{J}$, and only this one, makes all commutators polynomials in $\mathbf{J}$. In each sector of fixed energy and $J_z$, $(\mathbf{J}^2,T'_{zz})$ generates a cubic algebra of Heun--Racah type whose parameters are symmetric, up to a scalar shift, in the magnetic charge and the magnetic quantum number. At $q=0$, $T$ is the algebraic Heun operator of the pair (Casimir of $\mathfrak{so}(3)$, quadrupole of $\mathfrak{su}(3)$): the symmetry algebra is the $\mathfrak{so}(3)$-covariant extension of a Heun--Racah algebra. We define it abstractly and classify its multiplicity-free finite-dimensional representations at the energies of the system: besides the levels, they include, for nonzero integer $q$, a second family carried by the states the monopole removes, the kernel of the Newman--Penrose operator. All relations are derived from the dynamical algebra of the ladder operators and hold for integer and half-integer charges alike.

\end{abstract}

\section{Introduction}

A charged particle in the field of a Dirac magnetic monopole \cite{Dirac} is the simplest quantum system in which the wave function is not a function but a section of a line bundle \cite{WuYang}. The angular momentum acquires the component $-q\hat{\rr}$ along the line joining the particle to the monopole, and the orbital quantum number $l$ is bounded below by $|q|$, where $q$ is the product of the electric and magnetic charges. McIntosh and Cisneros \cite{MC} observed that when a central potential is added, the accidental degeneracies of the Coulomb and harmonic problems survive the monopole provided the repulsive centrifugal term $q^2/2r^2$ is added as well. For the Coulomb potential, the resulting system is the one studied by Zwanziger \cite{Zwanziger}. Its degeneracies are accounted for by an $O(4)$ symmetry, the Runge--Lenz vector being conserved, and its spectrum-generating algebra $SO(4,2)$ can be derived from the four-dimensional oscillator by dimensional reduction \cite{DHV_NP}. For the harmonic potential and integer $q$, the spectrum is that of the three-dimensional isotropic oscillator, but the level with energy $k+\tfrac32$ contains only the angular momenta $l=k,k-2,\dots$ that are at least $|q|$. The degeneracy is that of $SU(3)$ with the states of angular momentum $l<|q|$ removed from each level, and McIntosh and Cisneros noted that this degeneracy ``cannot be extracted from a symmetry group even though a calculus of ladder operators exists accounting for'' it \cite{MC}. The symmetries of this harmonic system were studied by Labelle, Mayrand and Vinet \cite{LMV}. They showed that a symmetric tensor $R_{ij}$ of rank~2, quartic in the velocities and built from the raising and lowering operators of the energy, commutes with the Hamiltonian, that no symmetry of order two or three in the velocities exists besides the angular momentum, and that the commutator $[R_{ij},R_{kl}]$ has a leading term proportional to $\epsilon_{ikm}HL_mR_{jl}$ followed by terms of lower degree, so that the commutators do not close linearly on $L$ and $R$. The relations of the algebra generated by $L$ and $R$ were not obtained, and each level was described instead by an $SU(2)\otimes SU(2)$ whose generators are the ladder operators of the quantum numbers, suitably normalised. Related classical and quantum hidden symmetries have also been studied by Inzunza, Plyushchay and Wipf \cite{IPW2020}.

The purpose of the present paper, thirty-five years after \cite{LMV}, is to obtain the relations of the algebra generated by $\J$ and $R$ and to identify it. The point of view is that of the polynomial symmetry algebras of superintegrable systems \cite{GVZ_LMP,Zhedanov_AW,GZL,KMP}. For systems of second order, these algebras are quadratic and their representations are governed by the orthogonal polynomials of the Askey scheme: the Racah algebra for the generic model on the two-sphere \cite{GVZ_LMP,KMP}, the Hahn algebra for the singular oscillators, which by a redefinition of generators is the cubic Higgs algebra of the oscillator on the sphere \cite{Higgs,Zhedanov_Higgs,FGVVZ}. Polynomial algebras of systems with monopoles have been obtained by Marquette for the generalised MICZ--Kepler system and its four-dimensional dual \cite{Marquette2010}, and by Hoque, Marquette and Zhang for Kepler-type systems with a $U(1)$ monopole in Taub--NUT space \cite{HMZ2016} and with the Yang--Coulomb $SU(2)$ monopole in five dimensions \cite{HMZ2017}. See also \cite{Meng} for the MICZ--Kepler problems in all dimensions. Closest to the present work is \cite{HMZ2018}, where the recurrence approach is applied to the Hamiltonian studied here and to a harmonic monopole system in a generalised Taub--NUT space. There, in each sector of fixed energy and magnetic quantum number, integrals of motion are built from the ladder operators of the radial quantum number and the shift operators of the angular momentum, normalised with $\sqrt{\mathbf L^2+\tfrac14}$, and they generate a polynomial algebra in deformed-oscillator form whose structure function yields the spectrum. These integrals are, up to normalisation, the parts of the tensor of \cite{LMV} that raise and lower $l$ by two units within a sector. Here we work with the tensor itself, whose components are polynomial in the velocities, the coordinates and the angular momentum, and with the full algebra it generates with $\J$, covariant under rotations. The sector algebra appears inside the centraliser of $J_z$, and we identify it. Section~\ref{sec:hmzcompare} makes the correspondence with the recurrence generators of \cite{HMZ2018} explicit. For nonzero monopole charge there is no separation of variables in Cartesian coordinates and the symmetry is of fourth order: what replaces the quadratic algebras is the object we describe.

Our results are the following. The polynomial algebra of the symmetries built from the ladder operators is generated by $\J$ and by the traceless symmetric tensor $T$ (the tensor $R$ of \cite{LMV} up to normalisation), and the other bilinear constants of motion, built from two raising operators and the radial lowering operator, reduce to $\J$ and $T$. The relations are written in covariant form: the rank-3 part of the commutator $[T,T]$ is $8(H-1)[J\otimes T]^{(3)}+8\Fi J^{(3)}$, its rank-1 part is a polynomial in $\J^2$ and $H$ times $\J$, and there is a symmetric relation of rank~2 which plays the role that Elliott's relation plays for $\su(3)$. Here $\Fi=2\J^2-2H+3-2q^2$. The algebra is not quadratic. The situation is best understood by the shift $T'=T+2(H-1)J^{(2)}$, where $J^{(2)}$ is the quadrupole of $\J$: the commutators of $T'$ are polynomials in $\J$ alone, with coefficients even in $H$. This shift is the analogue, one rank up, of the redefinition that identifies the Higgs algebra with the Hahn algebra. Here it does not produce a quadratic algebra, but it exhibits the structure, and it is the only shift by the quadrupole that does so. When one fixes the energy and the magnetic quantum number $m$, the operator $T'_{zz}$ is, in every sector, an algebraic Heun operator of a Hahn pair whose first member is $\J^2$, and $A=\J^2$, $B=T'_{zz}$ satisfy
\[
[A,[A,B]]=8\{A,B\}-12B+a_5A+a_8,\qquad [B,[A,B]]=-128A^3-8B^2+a_3A^2-a_5B+a_6A+a_9 ,
\]
which are the relations of the algebra $\mathcal A_{\rm gen}$ of \cite{CPV_CY} with $a_0=0$, that is, the relations satisfied by a Racah generator and a Heun--Racah operator \cite{BCTVZ}. The parameters $a_i$ are polynomials in $H$, $m^2$ and $q^2$ and become symmetric under $m^2\leftrightarrow q^2$ after an additive scalar shift of the sector generator. The origin of this Heun structure is visible at $q=0$, where the tensor $\tilde T=T'+2HJ^{(2)}$ is $-2\{\J^2-1,D\}$ with $D$ the quadrupole of $\su(3)$: the tensor of \cite{LMV} is the algebraic Heun operator \cite{GVZ_Heun} of the pair (Casimir operator of $\so(3)$, quadrupole of $\su(3)$), in tensor form, and its contraction $\J\cdot\tilde T\cdot\J$ is, up to the factor $-4(\J^2-1)$, the labelling operator of Bargmann and Moshinsky \cite{BM} for the chain $\su(3)\supset\so(3)$. The symmetry algebra of the oscillator in the monopole field is thus the $\so(3)$-covariant extension of a Heun--Racah algebra, in which the two labels $(\J^2,J_z)$ play the role of the Racah generator and the tensor $T'$ that of the Heun operator. We define this covariant algebra $\Alg(H,q)$ by generators and relations, and determine its finite-dimensional representations at the energies of the system, with angular momenta of the parity of the level and each occurring at most once. The reduced matrix elements of $T'$ are fixed by the relations, and the representations are the chains of angular momenta between consecutive zeros of a single polynomial, the lowest angular momenta $l=0,\tfrac12,1$ requiring a separate check. Among them are the energy levels of the physical system, whose matrix elements are all recovered, and, for nonzero integer $q$, a second family carried by the angular momenta $l<|q|$: the states that the monopole removes from the oscillator level, which are the kernel of the Newman--Penrose operator $\eth^{|q|}$, form representations of the same algebra. All the relations are derived from the dynamical algebra of the ladder operators, which is established by operator identities and determines their action on the states. The covariant relations follow from the dynamical algebra by reordering, without reference to the basis of states, and the Heun--Racah relations of the sectors follow from the covariant relations by the recoupling of $\so(3)$, and they hold for half-integer as well as integer charges. Reduced by the Wigner--Eckart theorem, the relations become identities between rational functions of the quantum numbers, which is the form in which the representations are classified. The degeneracy which McIntosh and Cisneros found could not be extracted from a symmetry group is thus extracted from a polynomial symmetry algebra: each level is irreducible under $\J$ and $T$, and the algebra is a Heun--Racah algebra made covariant under rotations.

The paper is organised as follows. Section~\ref{sec:model} recalls the system and its ladder operators and gives their algebra. Section~\ref{sec:algebra} contains the relations of the symmetry algebra, for the tensor of \cite{LMV} and for the shifted tensor. Section~\ref{sec:sectors} describes the reduction to the sectors of fixed energy and magnetic quantum number and the Heun operators found there. Section~\ref{sec:heun} establishes the Heun--Racah structure of the sector algebras and explains its origin at $q=0$. Section~\ref{sec:reps} defines the abstract algebra and classifies its representations. Section~\ref{sec:hmzcompare} identifies the recurrence operators within the tensor algebra. Section~\ref{sec:deformed} gives the cubic Casimir and a deformed-oscillator realisation of the sector algebra. Section~\ref{sec:proofs} gives the detailed proofs. Concluding remarks are collected in Section~\ref{sec:conclusion}, and the conventions in the Appendix.

\section{The model and its ladder operators}\label{sec:model}

We take $\hbar=\omega=1$ and unit mass. Operator identities are understood on the span of the regular separated eigenfunctions. Let $\mathbf A$ be the vector potential of a monopole of strength $g$ at the origin, $\nabla\times\mathbf A=g\,\rr/r^3$, defined by patches \cite{WuYang}, and let $\pp=\mathbf p-e\mathbf A$. With $q=eg\in\tfrac12\mathbb Z$, the Hamiltonian and the angular momentum are
\begin{equation}\label{eq:H}
H=\tfrac12\pp^2+\frac{q^2}{2r^2}+\tfrac12r^2,\qquad \J=\rr\times\pp-q\,\hat{\rr},\qquad [J_i,J_j]=i\epsilon_{ijk}J_k,\quad [H,J_i]=0 .
\end{equation}
The wave functions are sections, and in spherical coordinates they separate as $\phi_{nl}(r)Y_{q,l,m}(\theta,\varphi)$ with $Y_{q,l,m}$ the monopole harmonics \cite{WuYang}, $l\ge|q|$, $l-q\in\mathbb Z$, and $\phi_{nl}$ the radial functions of the three-dimensional oscillator, since the term $q^2/2r^2$ cancels the shift $-q^2$ of the angular eigenvalue $\J^2-q^2=l(l+1)-q^2$. Hence
\begin{equation}\label{eq:spectrum}
H|n,l,m\rangle=(k+\tfrac32)|n,l,m\rangle,\qquad k=2n+l,\quad n\in\mathbb N,\quad l\ge|q|,\quad -l\le m\le l ,
\end{equation}
where $\mathbb N=\{0,1,2,\ldots\}$, and the level $k$ carries the angular momenta $l=k,k-2,\dots$ down to $|q|$ or $|q|+1$. Its degeneracy is $\tfrac12(k+|q|+1)(k-|q|+2)$ when $k-|q|$ is even and $\tfrac12(k-|q|+1)(k+|q|+2)$ when it is odd \cite{LMV}. For integer $q$, in both cases it is the oscillator degeneracy $\tfrac12(k+1)(k+2)$ diminished by the number $\sum_{l<|q|,\,l\equiv k}(2l+1)$ of states with $l<|q|$.

Ladder operators for the energy were constructed in \cite{LMV}. With
\begin{equation}\label{eq:u}
u_i=\tfrac12\,\epsilon_{ijk}\{J_j,r_k\},
\end{equation}
which is a vector operator, one finds that $u$ has only components that raise or lower $H$ by one unit. We take as generators
\begin{equation}\label{eq:bB}
b_i^\dagger=u_i+[H,u_i]=\tfrac12\epsilon_{ijk}\{J_j,\,r_k-i\pi_k\},\quad b_i=(b_i^\dagger)^\dagger,\quad B_-=r^2+\tfrac12[H,r^2]-H,\quad B_+=(B_-)^\dagger ,
\end{equation}
so that $b_i^\dagger$ is twice the raising part of $u_i$ and $B_-$ is twice the part of $r^2$ that raises $H$ by two units, the part of $r^2$ that commutes with $H$ being $H$ itself. Our $b$ is $\sqrt2$ times the $b$ of \cite{LMV}. The operators $B_\pm$ and $H$ generate the $\so(2,1)$ of the radial problem \cite{Jackiw,DHV_CMP}. The following relations hold, with
\begin{equation}\label{eq:defs}
\Fi=2\J^2-2H+3-2q^2,\qquad \Cc=(H-\tfrac12)(H-\tfrac32)-\J^2,\qquad \Aa=\J^2+1-q^2 .
\end{equation}

\begin{proposition}\label{prop:dyn}
The operators $H$, $\J$, $b$, $b^\dagger$, $B_\pm$ satisfy
\begin{gather}
[H,b_i]=-b_i,\qquad [H,B_\pm]=\mp2B_\pm,\qquad [B_+,B_-]=4H,\qquad B_-B_+=\Cc,\label{eq:dyn1}\\
[B_+,b_i]=0,\qquad [B_-,b_i]=-2b_i^\dagger,\qquad \J\cdot b=b\cdot\J=0,\label{eq:dyn2}\\
[b_i,b_j]=-2i\epsilon_{ijk}J_kB_+,\qquad [b_i,b_j^\dagger]=2(2\J^2+1-q^2)\delta_{ij}-\{J_i,J_j\}-2i\epsilon_{ijk}J_kH,\label{eq:dyn3}\\
b\cdot b=2\Aa\,B_+,\qquad b^\dagger\!\cdot b=\Aa\,(2H-3)-2\J^2 ,\label{eq:dyn4}
\end{gather}
together with the adjoint relations.
\end{proposition}

At $q=0$, $B_+=a\cdot a$ and $b_i^\dagger$ is the raising part of $\epsilon_{ijk}\{L_j,x_k\}$, so that these relations are those of the Schr\"odinger algebra of the ordinary oscillator \cite{Niederer72,Niederer73,Hagen}, $(\so(2,1)\oplus\so(3))\ltimes\mathfrak h_3$, in which the Heisenberg brackets of the vector $(b,b^\dagger)$ have become the quadratic expressions \eqref{eq:dyn3}--\eqref{eq:dyn4}. The action of the generators on the basis $|n,l,m\rangle$ was computed in \cite{LMV} (eqs.~(4.7)--(4.8) there):
\begin{equation}\label{eq:action}
\begin{aligned}
b_z^\dagger|n,l,m\rangle&=c_0(n,l,m)\,|n,l+1,m\rangle+c_1(n,l,m)\,|n+1,l-1,m\rangle,\\
c_0(n,l,m)^2&=\frac{2\,(2n+2l+3)(l+1-m)(l+1+m)(l+1-q)(l+1+q)}{(2l+1)(2l+3)},\\
c_1(n,l,m)^2&=\frac{4\,(n+1)(l-m)(l+m)(l-q)(l+q)}{(2l-1)(2l+1)},
\end{aligned}
\end{equation}
the other components following by covariance (Appendix). Amplitudes to absent states are set to zero before evaluating the formulas, including $c_1$ at $l=|q|=\tfrac12$. In the phase convention $c_0,c_1\ge0$,
\begin{equation}\label{eq:Baction}
B_+|n,l,m\rangle=2\sqrt{n\,(n+l+\tfrac12)}\,|n-1,l,m\rangle .
\end{equation}
The factors $(l+1\pm q)$ in $c_0$ and $(l\pm q)$ in $c_1$ are the only trace of the monopole in these amplitudes, and they vanish on the border $l=|q|$ and thus prevent the ladder operators from reaching the states with $l<|q|$. Proposition~\ref{prop:dyn} is proved in Section~\ref{sec:proofdyn} by operator identities. Conversely, its relations determine \eqref{eq:action}--\eqref{eq:Baction} on the representation generated from the ground state (Proposition~\ref{prop:rep} below), which gives a derivation of the action computed in \cite{LMV} from the algebra alone. Everything that follows is a consequence of the dynamical algebra, and the symmetry algebra is derived from it by reordering in Section~\ref{sec:reorder}.

\section{The symmetry algebra}\label{sec:algebra}

\subsection{Generators and relations}
The operators that commute with $H$ and are bilinear in the ladder operators are $\J$, $P_{ij}=b_i^\dagger b_j$, $Q_{ij}=b^\dagger_ib^\dagger_jB_+$ and $Q_{ij}^\dagger$. Their antisymmetric parts and traces reduce to $\J$ and $H$:
\begin{equation}\label{eq:PQred}
P_{ij}-P_{ji}=i\epsilon_{ijk}J_k\,\Fi,\qquad P_{kk}=2\Cc+(H-\tfrac32)\Fi,\qquad Q_{kk}=2\Aa\,\Cc .
\end{equation}
We retain the traceless symmetric parts
\begin{equation}\label{eq:TS}
T_{ij}=P_{ij}+P_{ji}-\tfrac23\delta_{ij}P_{kk},\qquad S_{ij}=Q_{ij}+Q_{ji}-\tfrac23\delta_{ij}Q_{kk},
\end{equation}
and $S^\dagger$. The tensor $T$ is, up to normalisation and a multiple of $\delta_{ij}$, the conserved tensor $R_{ij}$ of \cite{LMV}. On the basis \eqref{eq:spectrum}, $\Fi=2[l(l+1)-q^2-k]$ and $\Cc=(k-l)(k+l+1)=2n(k+l+1)$.

\paragraph{Conventions.} The spherical components of a vector $V$ are $V_0=V_z$, $V_{\pm1}=\mp(V_x\pm iV_y)/\sqrt2$, and those of a symmetric traceless tensor $T_{ij}$ are
\[
T_0=\sqrt{\tfrac32}\,T_{zz},\qquad T_{\pm1}=\mp(T_{xz}\pm iT_{yz}),\qquad T_{\pm2}=\tfrac12(T_{xx}-T_{yy}\pm2iT_{xy}) ,
\]
and they are the components of an irreducible tensor operator of rank 2, $[J_z,T_\mu]=\mu T_\mu$, $[J_\pm,T_\mu]=\sqrt{(2\mp\mu)(3\pm\mu)}\,T_{\mu\pm1}$. Two irreducible tensor operators $A$ and $B$ of ranks $j_1$ and $j_2$ are coupled to rank $L$ by
\begin{equation}\label{eq:coupling}
[A\otimes B]^{(L)}_M=\sum_{m_1+m_2=M}\langle j_1m_1\,j_2m_2|LM\rangle\,A_{m_1}B_{m_2},
\end{equation}
$B$ acting first. The tensors built from $\J$ are $J^{(1)}=\J$, the quadrupole $J^{(2)}=[J\otimes J]^{(2)}$, whose Cartesian components are $J_{(ij)}=\tfrac12\{J_i,J_j\}-\tfrac13\delta_{ij}\J^2$, and $J^{(3)}=[J\otimes J^{(2)}]^{(3)}$. When the same tensor $T$ is coupled to itself, the symmetry $\langle2\mu\,2\nu|LM\rangle=(-1)^L\langle2\nu\,2\mu|LM\rangle$ of the Clebsch--Gordan coefficients gives
\begin{equation}\label{eq:TTdef}
[T\otimes T]^{(L)}_M=\tfrac12\sum_{\mu+\nu=M}\langle2\mu\,2\nu|LM\rangle\,[T_\mu,T_\nu]_{\pm}\,,\qquad [\,\cdot,\cdot\,]_\pm=\begin{cases}[\,\cdot,\cdot\,]& L=1,3,\\ \{\,\cdot,\cdot\,\}& L=0,2,4,\end{cases}
\end{equation}
so that $[T\otimes T]^{(1)}$ and $[T\otimes T]^{(3)}$ are the vector and rank-3 parts of the commutator $[T_\mu,T_\nu]$, which they determine completely,
\begin{equation}\label{eq:commcoupled}
[T_\mu,T_\nu]=2\sum_{L=1,3}\langle2\mu\,2\nu|LM\rangle\,[T\otimes T]^{(L)}_M ,
\end{equation}
while $[T\otimes T]^{(0)}$, $[T\otimes T]^{(2)}$, $[T\otimes T]^{(4)}$ are the parts of the anticommutator. For instance $[T\otimes T]^{(3)}_{3}=[T_2,T_1]/\sqrt2$, and $\sqrt5\,[T\otimes T]^{(0)}=\sum_\mu(-1)^\mu T_\mu T_{-\mu}=T_{ij}T_{ji}=\Tr T^2$. The contraction $\J\cdot T$ denotes the vector $J_iT_{ij}$, and in coupled form $(\J\cdot T)_\mu=-\sqrt{5/3}\,[J\otimes T]^{(1)}_\mu$.

\begin{theorem}\label{thm:T}
With $H$ central, the angular momentum $\J$ and the tensor $T$ satisfy
\begin{align}
\J\cdot T&=T\cdot\J=-\tfrac23\,(2\Cc+H\Fi)\,\J, \label{eq:transv}\\
[T\otimes T]^{(3)}&=8(H-1)\,[J\otimes T]^{(3)}+8\,\Fi\,J^{(3)}, \label{eq:TT3}\\
[T\otimes T]^{(1)}&=\sqrt{\tfrac{18}{5}}\;\chi\,J,\qquad
3\chi=5\Fi^2-16\Fi\,(\Cc+H-\tfrac52)-16\,\Cc\,(2H+3), \label{eq:TT1}\\
[T\otimes T]^{(2)}&=\tfrac{4}{\sqrt{21}}\Big[(1-H)\{\J^2,T\}+\big(\tfrac92-(5-2q^2)H\big)T+\big(24\,\Cc-\tfrac32\Fi(\Fi+12)\big)J^{(2)}\Big], \label{eq:TT2}\\
\Tr T^2&=-\tfrac{64}{3}\Cc^2+\Cc(32H^2+56)+\Fi^2\big(2\Cc+\tfrac23H^2+4H-\tfrac{15}{2}\big)\notag\\
&\quad+\Fi\big(\tfrac{32}{3}\Cc H+48\Cc-16H^2+64H-60\big). \label{eq:trT2}
\end{align}
\end{theorem}

\begin{proposition}\label{prop:S}
$S=(H-2)\,T+\tfrac14[\J^2,T]+\Fi\,J^{(2)}$ and $S^\dagger=(H-2)\,T-\tfrac14[\J^2,T]+\Fi\,J^{(2)}$. The polynomial algebra of the symmetries of $H$ built from the ladder operators is therefore generated by $\J$ and $T$.
\end{proposition}

Relation \eqref{eq:TT3} is the commutator $[R_{ij},R_{kl}]$ of \cite{LMV}, whose leading term $\propto\epsilon_{ikm}HL_mR_{jl}$ was given there, now complete: the terms of lower degree are $8\Fi J^{(3)}$ and the $H$-independent part of $8(H-1)[J\otimes T]^{(3)}$. The shape of the relations can be read off the dynamical algebra. For the bilinears $P_{ij}=b_i^\dagger b_j$ one has the identity
\begin{equation}\label{eq:PPcomm}
[P_{ij},P_{kl}]=b_i^\dagger[b_j,b_k^\dagger]b_l-b_k^\dagger[b_l,b_i^\dagger]b_j-2i\epsilon_{jlm}\,b_i^\dagger b_k^\dagger J_mB_+-2i\epsilon_{ikm}\,B_-J_mb_lb_j ,
\end{equation}
where \eqref{eq:dyn3} has been used in the last two terms. The first two terms are linear in $P$ with coefficients polynomial in $\J$ and $H$, while the last two are of the form $Q\cdot J$ and $J\cdot Q^\dagger$, and would not be there for the ordinary oscillator, where $[b_i,b_j]=0$. Proposition~\ref{prop:S}, which expresses $S$ and $S^\dagger$ through $T$, is what closes the algebra on $\J$ and $T$, and the two contributions are visible in \eqref{eq:TT3}: $8(H-1)[J\otimes T]^{(3)}$ comes from the $Q$ terms through Proposition~\ref{prop:S} and from the $[b,b^\dagger]$ terms, $8\Fi J^{(3)}$ from the antisymmetric part of $P$ and from Proposition~\ref{prop:S} in equal parts. The derivation is carried out in Section~\ref{sec:reorder}, where \eqref{eq:PQred}, Proposition~\ref{prop:S} and Theorem~\ref{thm:T} are obtained from Proposition~\ref{prop:dyn} by reordering. The transversality \eqref{eq:transv} follows from $\J\cdot b=b\cdot\J=0$ alone, given \eqref{eq:PQred}: covariance gives $J_iP_{ij}=(b\cdot\J)^\dagger b_j=0$ and $J_iP_{ji}=b_j^\dagger(\J\cdot b)+i\epsilon_{ijk}P_{ki}=i\epsilon_{ijk}P_{ki}$, hence $J_iT_{ij}=\tfrac{i}{2}\epsilon_{ijk}(P_{ki}-P_{ik})-\tfrac23J_jP_{kk}=-(\Fi+\tfrac23P_{kk})J_j$. Relations \eqref{eq:TT2} and \eqref{eq:trT2} are the analogues of Elliott's relation \cite{Elliott} and of the Casimir relation in the symmetric representations of $\su(3)$, to which they reduce at $q=0$ once $T$ is expressed through the quadrupole (Section~\ref{sec:heun}). Each energy level is irreducible under $\langle\J,T\rangle$: $T$ connects $l$ to $l$ and $l\pm2$ only, and none of the blocks $l\to l+2$ vanishes inside a level, by \eqref{eq:5pairs} below.

\subsection{The shifted tensor}
Let
\begin{equation}\label{eq:shift}
T'=T+2(H-1)\,J^{(2)},\qquad \tilde T=T+2(2H-1)\,J^{(2)}=T'+2H\,J^{(2)} .
\end{equation}

\begin{theorem}\label{thm:Tp}
The commutators of $T'$ are polynomials in $\J$ alone, even in $H$:
\begingroup\allowdisplaybreaks[4]
\begin{align}
\J\cdot T'&=T'\cdot\J=\tfrac13(4q^2-1)\,H\,\J,\label{eq:transvp}\\
[T'\otimes T']^{(3)}&=8\,(2\Aa-H^2)\,J^{(3)},\label{eq:TT3p}\\
[T'\otimes T']^{(1)}&=\sqrt{\tfrac{18}{5}}\,f_1\,J,\label{eq:TT1p}\\
3f_1&=52\Aa^2+(104+32q^2)\Aa-36H^2\Aa-(7-4q^2)H^2-60(1-q^2),\notag\\
[T'\otimes T']^{(2)}&=\tfrac{4}{\sqrt{21}}\Big[\tfrac12(4q^2-1)H\,T'+p_2\,J^{(2)}\Big],\label{eq:TT2p}\\
2p_2&=-12\Aa^2+12(H^2-10)\Aa+(33-12q^2)H^2-36(q^2-1),\notag\\
\Tr T'^2&=-8\lambda^3+(16q^2-168)\lambda^2+(168q^2-8q^4-304)\lambda-24q^4+120q^2-96\notag\\
&\quad+H^2\big(8\lambda^2+(78-16q^2)\lambda+\tfrac{32}{3}q^4-\tfrac{160}{3}q^2+\tfrac{128}{3}\big),\qquad \lambda=\J^2,\label{eq:trTp2}\\
[\J^2,[\J^2,T']]&=8\{\J^2,T'\}-12\,T'+8(1-4q^2)H\,J^{(2)},\label{eq:HR1}
\end{align}
\endgroup
and $[\tilde T\otimes\tilde T]^{(3)}=-8H\,[J\otimes\tilde T]^{(3)}+16\,\Aa\,J^{(3)}$.
\end{theorem}

Theorems~\ref{thm:T} and \ref{thm:Tp} are equivalent, the passage from one to the other being the two $\so(3)$ identities
\begin{equation}\label{eq:shiftid}
[T\otimes J^{(2)}]^{(3)}+[J^{(2)}\otimes T]^{(3)}=-4\,[J\otimes T]^{(3)},\qquad [J^{(2)}\otimes J^{(2)}]^{(3)}=-2\,J^{(3)},
\end{equation}
valid for every rank-2 tensor operator $T$, together with their analogues of ranks $0$, $1$, $2$ (Lemma~\ref{lem:so3}). Theorem~\ref{thm:Tp} is deduced from Theorem~\ref{thm:T} in this way in Section~\ref{sec:reorder}. Relation \eqref{eq:HR1} is the general identity \eqref{eq:HRgen} satisfied by the double commutator of $\J^2$ with a rank-2 tensor operator, combined with \eqref{eq:transvp}. The shift \eqref{eq:shift} is the analogue, one rank up, of the redefinition of generators that identifies the cubic Higgs algebra with the quadratic Hahn algebra \cite{Zhedanov_Higgs,FGVVZ}: here it removes the term $[J\otimes T]^{(3)}$ from the rank-3 relation and sets $a_2=0$ in the sector algebra of Section~\ref{sec:heun}. It is the only shift with a scalar coefficient depending on $H$ that does so.

\begin{proposition}\label{prop:shift}
Among the tensors $T_\beta=T+\beta\,J^{(2)}$ with $\beta$ a function of $H$, the rank-3 relation of $T_\beta$ contains no term $[J\otimes T_\beta]^{(3)}$ if and only if $\beta=2(H-1)$, i.e.\ $T_\beta=T'$.
\end{proposition}
\begin{proof}
By \eqref{eq:TT3} and \eqref{eq:shiftid}, $[T_\beta\otimes T_\beta]^{(3)}=\big(8(H-1)-4\beta\big)[J\otimes T]^{(3)}+(8\Fi-2\beta^2)J^{(3)}$, and $[J\otimes T_\beta]^{(3)}=[J\otimes T]^{(3)}+\beta J^{(3)}$ since $[J\otimes J^{(2)}]^{(3)}=J^{(3)}$. Hence
\[
[T_\beta\otimes T_\beta]^{(3)}=\big(8(H-1)-4\beta\big)[J\otimes T_\beta]^{(3)}+\big(8\Fi-8\beta(H-1)+2\beta^2\big)J^{(3)} ,
\]
whose first coefficient vanishes exactly for $\beta=2(H-1)$, and the second is then $8\Fi-8(H-1)^2=8(2\Aa-H^2)$, which is \eqref{eq:TT3p}.
\end{proof}

Since the rank-1 relation of $T_\beta$ is a multiple of $\J$ for every $\beta$, by \eqref{eq:transv} and Lemma~\ref{lem:so3}(iii), $T'$ is also the only shift by the quadrupole for which all commutators are polynomials in $\J$ alone. One cannot go further. A term $J^{(3)}$ with a coefficient depending on $\J^2$ cannot be removed by any polynomial redefinition $T'\to\alpha T'+\beta(\J^2,H)J^{(2)}$, since $[J^{(2)}\otimes J^{(2)}]^{(3)}$ has a constant coefficient. In particular the algebra is not quadratic, for a quadratic algebra would give, after such a shift, a constant coefficient in front of $J^{(3)}$, whereas \eqref{eq:TT3p} has the coefficient $8(2\Aa-H^2)$, linear in $\J^2$.

\section{The sectors of fixed energy and magnetic quantum number}\label{sec:sectors}

What $T'$ explains, and $\J$ alone does not, is the degeneracy between the different angular momenta $l=k,k-2,\dots$ of one level: $\J$ preserves $l$, whereas $T'$ connects $l$ to $l\pm2$. This is seen most directly in the sectors where the two conserved labels $H$ and $J_z$ are fixed, in which $T'_{zz}$ acts on the finite set of angular momenta of the level.

In the sector $(k,m)$ spanned by the $|n,l,m\rangle$ with $2n+l=k$, $l=l_0,l_0+2,\dots,k$, where $l_0$ is the smallest $l\ge\max(|m|,|q|)$ of the parity of $k$, $\J^2$ is diagonal and $T_{zz}$, $T'_{zz}$ are tridiagonal, with the same off-diagonal entries. With $M_+=\max(|m|,|q|)$, $M_-=\min(|m|,|q|)$ and $\langle c\rangle_l=(l+1+c)(l+2-c)$,
\begin{equation}\label{eq:5pairs}
(T_{zz})_{l,l+2}\,(T_{zz})_{l+2,l}=16\;\frac{\langle M_+\rangle_l\langle -M_+\rangle_l\langle M_-\rangle_l\langle -M_-\rangle_l\;(k-l)(k+l+3)}{(2l+1)(2l+3)^2(2l+5)} ,
\end{equation}
and
\begin{equation}\label{eq:dl}
(T'_{zz})_{ll}=-\tfrac23(4q^2-1)H\,\frac{l(l+1)-3m^2}{(2l-1)(2l+3)} .
\end{equation}
The products \eqref{eq:5pairs} are symmetric under $m^2\leftrightarrow q^2$, but the diagonal $d_l(q,m)=(T'_{zz})_{ll}$ satisfies
\[
d_l(q,m)-d_l(m,q)=\tfrac23H(m^2-q^2).
\]
Thus the sector matrices become symmetric under this exchange after replacing $T'_{zz}$ by $T'_{zz}+\tfrac23Hq^2\mathbf1$. The wave functions have the corresponding local symmetry. In the gauge of \cite{WuYang} the monopole harmonics are $Y_{q,l,m}=\sqrt{(2l+1)/4\pi}\,e^{i(m+q)\varphi}d^l_{m,-q}(\theta)$ (Section~\ref{sec:eth}), and $d^l_{m,-q}=d^l_{q,-m}$, so that $Y_{q,l,m}=Y_{m,l,q}$ as local functions: the sector $(k,m)$ of the system of charge $q$ and the sector $(k,q)$ of the system of charge $m$ are spanned by the same functions $\phi_{nl}(r)Y_{q,l,m}$, and the two operators $T'_{zz}$ have matrices differing by the scalar above. This is not a global gauge equivalence between the different monopole bundles. In terms of $s=l+\tfrac32$, each pair $\langle c\rangle_l$ is $s^2-(c-\tfrac12)^2$, so that the numerator of \eqref{eq:5pairs} is a product of five factors $s^2-c_i^2$, with $\{c_i\}=\{M_+\pm\tfrac12,\,M_-\pm\tfrac12,\,H\}$, against three for the Jacobi matrix of a Hahn polynomial and four for a Racah polynomial. The products \eqref{eq:5pairs} carry the same information as the structure function of the deformed-oscillator algebra of \cite{HMZ2018}. The relation with the Hahn polynomials is as follows. Let $Y=K_0^{(2)}$ be the diagonal generator of the second factor in the Clebsch--Gordan problem $D^+(k_1)\otimes D^+(k_2)$ of $\su(1,1)$, written in the coupled basis of the total representations $D^+(k_1+k_2+j)$, $j=0,\dots,N$ \cite{VdJ,GVZ_LMP}. It is a tridiagonal matrix whose entries are given in Section~\ref{sec:proofs}. With $j=(l-l_0)/2$, the pair $(\J^2,Y)$ is a Hahn pair and generates a Hahn algebra.

\begin{proposition}\label{prop:heunhahn}
In every sector $(k,m)$, let $N=(k-l_0)/2$, $2(k_1+k_2)=l_0+\tfrac32$, and let $2|k_1-k_2|$ be any of the four numbers $M_+\pm\tfrac12$, $M_-\pm\tfrac12$ other than $l_0-\tfrac12$, and call $a$ and $b$ the two remaining ones and $\varepsilon=-\operatorname{sign}(k_1-k_2)$. Then there is a diagonal matrix $D$ such that
\begin{gather}\label{eq:heunhahn}
T'_{zz}=D\,\big(\tau_1\,\J^2Y+\tau_2\,Y\J^2-2H\,\J^2+\tau_4\,Y+\tau_0\big)\,D^{-1},\\
\tau_1+\tau_2=8,\qquad \tau_2-\tau_1=\pm4(a+\varepsilon b),\qquad \tau_4=-6+8\varepsilon ab ,\nonumber
\end{gather}
with $\tau_0$ a constant. Thus $T'_{zz}$ is, up to a diagonal similarity, the algebraic Heun operator of the Hahn pair $(\J^2,Y)$, for each of the admissible choices of $(k_1,k_2)$.
\end{proposition}

The similarity $D$ is needed because the bilinear form in \eqref{eq:heunhahn} is not symmetric when $\tau_1\neq\tau_2$, and it is a diagonal gauge whose successive ratios $D_{l+2}/D_l$ are square roots of rational functions of $l$. The symmetric form $\tau_1=\tau_2$ occurs exactly when $a+\varepsilon b=0$, that is when two of the five pairs coincide: this happens when $M_-=0$, when $|m|=|q|$, and when $|m|-|q|=\pm1$ with $l_0=M_++1$. There $T'_{zz}=4\{\J^2,Y\}-2H\J^2+\tau_4Y+\tau_0$ is a Heun--Hahn operator in the sense of \cite{VZ_HeunHahn}, and at $q=0$ every sector is of this kind, the pair being the spherical and cylindrical bases of the oscillator. In the other sectors the Heun operator is of the general bilinear form of \cite{GVZ_Heun}. That $T'_{zz}$ is a Heun operator of several Hahn pairs at once, up to six for a given sector, reflects the fact that the five pairs of \eqref{eq:5pairs} can be split into three Hahn pairs and two remaining ones in several ways.

\section{The Heun--Racah algebra of the sectors}\label{sec:heun}

\begin{theorem}\label{thm:HR}
In each sector $(k,m)$, with $A=\J^2$, $B=T'_{zz}$ and $C=[A,B]$,
\begin{align}
[A,C]&=8\{A,B\}-12\,B+a_5A+a_8,\label{eq:AC}\\
[B,C]&=-128\,A^3-8\,B^2+a_3A^2-a_5B+a_6A+a_9,\label{eq:BC}
\end{align}
with
\begin{gather*}
a_5=\tfrac83(4q^2-1)H,\qquad a_8=-8(4q^2-1)Hm^2,\qquad a_3=96H^2+192(m^2+q^2)-896,\\
a_6=-1024+704(m^2+q^2)-64(m^2+q^2)^2-128m^2q^2+304H^2-128H^2(m^2+q^2),\\
a_9=-256+320(m^2+q^2)-64(m^4+q^4)-384m^2q^2+64m^2q^2(m^2+q^2)\\
\qquad+\tfrac{1024}{9}H^2-\tfrac{1280}{9}H^2q^2+\tfrac{256}{9}H^2q^4-144H^2m^2+128H^2m^2q^2+32H^2m^4 .
\end{gather*}
\end{theorem}

Theorem~\ref{thm:HR} is the restriction to a sector of Proposition~\ref{prop:status}(ii) below, which establishes \eqref{eq:AC}--\eqref{eq:BC}, with $J_z^2$ in place of $m^2$, as consequences of the covariant relations of Theorem~\ref{thm:Tp} and of the recoupling identities of $\so(3)$. These are the defining relations of the algebra $\mathcal A_{\rm gen}$ of \cite{CPV_CY}, eq.~(1) there, with
\[
a_0=0,\qquad a_0'=-128,\qquad a_1=8,\qquad a_2=0,\qquad a_4=-12 ,
\]
and the parameters $a_3,a_5,a_6,a_8,a_9$ above, polynomials in the central elements $H$, $J_z$ and $q^2$. The nontrivial consistency is that $a_1$, $a_2$ and $a_5$ take the same value in \eqref{eq:AC} and \eqref{eq:BC}. The specialisation $a_0=0$ of $\mathcal A_{\rm gen}$ is, as noted in \cite{CPV_CY}, a particular case of the Heun--Racah algebra of \cite{BCTVZ}, generated by a Racah generator and a Heun--Racah operator, in agreement with Proposition~\ref{prop:heunhahn}. It also appears as the cubic symmetry algebra of superintegrable systems with third-order integrals \cite{Marquette}. Since $a_0=0$ and $a_0'\neq0$, the sector relations in these generators are neither the Racah algebra ($a_0'=a_3=0$) nor the cubic algebra of Lehrer and Racah for the missing label of $\su(3)\supset\so(3)$, which is the stratum $a_0=-6$, $a_0'=-2$ \cite{Racah,CPV_CY}. For the unshifted generator $B_0=T_{zz}=B+\tfrac23(H-1)A-2(H-1)m^2\mathbf1$, the relations have the same general cubic form, with $a_2=-\tfrac{32}{3}(H-1)$ and correspondingly transformed $a_3,a_5,a_6,a_8,a_9$. The coefficients of Theorem~\ref{thm:HR} become symmetric in $(m^2,q^2)$ under the scalar change $B\mapsto B+\tfrac23Hq^2\mathbf1$, so the two exchanged sectors are related by this affine change of generator. This duality between the magnetic charge and the magnetic quantum number is already visible in \eqref{eq:5pairs}.

\paragraph{Origin of the Heun operator.}
At $q=0$, let $D_{ij}$ be the quadrupole of $\su(3)$, the traceless symmetric part of $a_i^\dagger a_j$. Then
\begin{equation}\label{eq:q0}
\tilde T=-2\,\{\J^2-1,\;D\},\qquad\text{that is}\qquad T=-2\{\J^2,D\}+4D+2(1-2H)J^{(2)} .
\end{equation}
After the quadrupole shift in \eqref{eq:shift}, the conserved tensor of \cite{LMV} is therefore, at $q=0$, an algebraic Heun operator \cite{GVZ_Heun} of the pair formed by the Casimir operator of $\so(3)$ and the quadrupole of $\su(3)$: a particular bilinear expression in the two, written in tensor form. Its contraction $\J\cdot\tilde T\cdot\J=-4(\J^2-1)X_3$ reproduces the labelling operator $X_3=\J\cdot D\cdot\J$ of Bargmann and Moshinsky \cite{BM} for the chain $\su(3)\supset\so(3)$, which in the symmetric representations is a function of $(H,\J^2)$. In a sector, $D_{zz}=n_z-\tfrac13(H-\tfrac32)$ and $(\J^2,n_z)$ is the Hahn pair of Section~\ref{sec:sectors}, with $Y=\tfrac12(H-n_z-\tfrac12)$, and $T'_{zz}$ is its Heun--Hahn operator. For $q\neq0$ the Cartesian side of the picture, $D$ and the number operators $n_i$, disappears, but the Hahn pairs and the Heun operators of Proposition~\ref{prop:heunhahn} remain, $q^2$ entering only through their parameters.

\section{The covariant algebra and its representations}\label{sec:reps}
\subsection{The abstract algebra and its representations}
\begin{definition}\label{def:alg}
Let $H$ be central and $q^2$ a parameter. The algebra $\Alg(H,q)$ is generated by $\J$, satisfying the $\so(3)$ relations, and by a rank-2 tensor operator $T'$ with respect to $\J$, subject to \eqref{eq:transvp}, \eqref{eq:TT3p}, \eqref{eq:TT1p}, \eqref{eq:TT2p} and \eqref{eq:trTp2}.
\end{definition}

Two algebras must be distinguished: the concrete symmetry algebra $\mathcal S(H,q)$, the algebra of operators on the Hilbert space of the system generated by $\J$ and $T'$, and the abstract algebra $\Alg(H,q)$ presented by generators and relations. Their relation is as follows.

\begin{proposition}\label{prop:status}
\begin{enumerate}[label=(\roman*)]
\item The map sending the generators of $\Alg(H,q)$ to the operators $\J$, $T'$ of the system extends to a surjective homomorphism $\Alg(H,q)\to\mathcal S(H,q)$. Whether it is injective, i.e.\ whether \eqref{eq:transvp}--\eqref{eq:trTp2} is a complete presentation of $\mathcal S(H,q)$, is not known.
\item In $\Alg(H,q)$ the elements $A=\J^2$ and $B=T'_{zz}$ commute with $J_z$ and $H$, and satisfy the relations \eqref{eq:AC}--\eqref{eq:BC} of Theorem~\ref{thm:HR} with $m^2$ replaced by $J_z^2$.
\end{enumerate}
\end{proposition}

\begin{proof}
(i) is Theorem~\ref{thm:Tp}. (ii) Relation \eqref{eq:AC} is the $zz$ component of \eqref{eq:HR1}, which holds in $\Alg(H,q)$ by the identity \eqref{eq:HRgen} of Lemma~\ref{lem:so3} and \eqref{eq:transvp}. With $J_{(zz)}=J_z^2-\tfrac13\J^2$ it gives $a_5=\tfrac83(4q^2-1)H$ and $a_8=-8(4q^2-1)HJ_z^2$. For \eqref{eq:BC}, with $C=[A,B]=\sum_i\big(J_i[J_i,B]+[J_i,B]J_i\big)$, moving the $J_i$ through $B$ gives the identity
\[
[B,C]=\sum_i\Big(J_i\big[B,[J_i,B]\big]+\big[B,[J_i,B]\big]J_i\Big)-2\sum_i[J_i,B]^2 .
\]
Since $B=\sqrt{2/3}\,T'_0$ and $[J_{\pm1},B]=\mp\sqrt2\,T'_{\pm1}$, $[J_0,B]=0$, this reads
\begin{equation}\label{eq:BCred}
[B,C]=\sqrt{\tfrac23}\Big(\{J_+,[T'_0,T'_{-1}]\}+\{J_-,[T'_0,T'_1]\}\Big)-4\{T'_1,T'_{-1}\} .
\end{equation}
The commutators are given by \eqref{eq:commcoupled}, $[T'_0,T'_{\mp1}]=\mp2\sqrt{\tfrac{3}{10}}\,[T'\otimes T']^{(1)}_{\mp1}\pm\tfrac{2}{\sqrt5}\,[T'\otimes T']^{(3)}_{\mp1}$, hence by \eqref{eq:TT1p} and \eqref{eq:TT3p} they are polynomials in $\J$. With $J^{(3)}_{1}=\tfrac{\sqrt{30}}{20}\,J_+\Xi$, $J^{(3)}_{-1}=-\tfrac{\sqrt{30}}{20}\,\Xi J_-$, $\Xi=\J^2-5J_z^2-5J_z-2$, one finds $\{J_+,J_{-1}\}=-\{J_-,J_{+1}\}=\sqrt2(\J^2-J_z^2)$ and $\{J_+,J^{(3)}_{-1}\}=-\{J_-,J^{(3)}_{+1}\}=-\tfrac{\sqrt{30}}{10}\,\Omega$ with $\Omega=\J^4-6\J^2J_z^2-2\J^2+5J_z^4+7J_z^2$, so that the first term of \eqref{eq:BCred} is $-\tfrac{24}{5}f_1(\J^2-J_z^2)-\tfrac{32}{5}(2\Aa-H^2)\,\Omega$. The anticommutator and $B^2=\tfrac23T_0'^2$ decompose by \eqref{eq:TTdef} into the parts of rank $0$, $2$, $4$ of $[T'\otimes T']$. With $\langle2\,1\,2\,{-1}|L\,0\rangle=-\tfrac{1}{\sqrt5},\tfrac{1}{\sqrt{14}},\sqrt{\tfrac{8}{35}}$ and $\langle2\,0\,2\,0|L\,0\rangle=\tfrac{1}{\sqrt5},-\sqrt{\tfrac27},\sqrt{\tfrac{18}{35}}$ for $L=0,2,4$,
\[
-4\{T'_1,T'_{-1}\}+8B^2=\tfrac{8}{3}\Tr T'^2-\tfrac{4\sqrt{14}}{3}\,[T'\otimes T']^{(2)}_0 ,
\]
the parts of rank $4$ cancelling exactly, which is why the coefficient of $B^2$ in \eqref{eq:BC} is $-8$. By \eqref{eq:TT2p} and $J^{(2)}_0=\tfrac{1}{\sqrt6}(3J_z^2-\J^2)$, $-\tfrac{4\sqrt{14}}{3}[T'\otimes T']^{(2)}_0=-\tfrac83(4q^2-1)H\,B-\tfrac{16}{9}\,p_2\,(3J_z^2-\J^2)$. Collecting,
\[
[B,C]+8B^2+a_5B=-\tfrac{24}{5}\,f_1\,(A-J_z^2)-\tfrac{32}{5}\,(2\Aa-H^2)\,\Omega+\tfrac83\Tr T'^2-\tfrac{16}{9}\,p_2\,(3J_z^2-A) ,
\]
a polynomial in $A=\J^2$, $J_z^2$, $H$ and $q^2$, and direct expansion with the values of $f_1$, $p_2$ and $\Tr T'^2$ of Theorem~\ref{thm:Tp} shows that it equals $-128A^3+a_3A^2+a_6A+a_9$ with the $a_i$ of Theorem~\ref{thm:HR}, $m^2$ replaced by $J_z^2$. This proves \eqref{eq:BC} in $\Alg(H,q)$, and Theorem~\ref{thm:HR} follows through the homomorphism of (i).
\end{proof}

Thus the centraliser of $J_z$ in $\Alg(H,q)$ contains the image of a specialised Heun--Racah algebra, and $\Alg(H,q)$ is the $\so(3)$-covariant extension of a Heun--Racah algebra, the two labels $(\J^2,J_z)$ holding the place of the Racah generator and the tensor $T'$ that of the Heun operator.

Fix $q\in\tfrac12\mathbb Z$ and $k\ge|q|$ with $k-|q|\in\mathbb Z$, and let $V=\bigoplus_{l\in\Lambda}D_l$ be a finite-dimensional representation of $\Alg(H,q)$ in which $H=k+\tfrac32$, each angular momentum occurs at most once, and $\Lambda\subset k+2\mathbb Z$ is an arithmetic progression of step~2: the energy and the parity of the angular momenta are those of the level $k$ of the system. By the Wigner--Eckart theorem, $\langle l'm'|T'_\mu|lm\rangle=\langle l\,m\,2\,\mu|l'\,m'\rangle R_{l'l}$ with $R_{l'l}=0$ unless $l'\in\{l,l\pm2\}$, and everything reduces to $R_{ll}$ and $P_l=R_{l+2,l}R_{l,l+2}$. Set
\begin{equation}\label{eq:Pl}
P(l)=16\,\frac{\langle q\rangle_l\langle -q\rangle_l\,(k-l)(k+l+3)\,(l+1)(l+2)}{(2l+3)\sqrt{(2l+1)(2l+5)}} ,\qquad \langle c\rangle_l=(l+1+c)(l+2-c),
\end{equation}
where the numerator has zeros at $l=-1,-2$, $l=\pm|q|-1,\pm|q|-2$ and $l=k,-k-3$. At the lower endpoints $l_0\in\{0,1\}$, $P(l_0-2)$ stands for the rational form $\Pc_{l_0-2}$ of \eqref{eq:Pcl}, whose numerator vanishes there.

\begin{theorem}\label{thm:reps}
\begin{enumerate}
\item Relation \eqref{eq:transvp} fixes $R_{ll}=\tfrac23(4q^2-1)H\sqrt{\dfrac{3\,l(l+1)}{2(2l-1)(2l+3)}}$ for $l\ge1$, and $R_{ll}=0$ for $l\in\{0,\tfrac12\}$.
\item For each $l\ge2$ in $\Lambda$, the diagonal blocks of \eqref{eq:TT3p} and \eqref{eq:TT2p} form a linear system in $(P_l,P_{l-2})$ whose determinant, proportional to $\sqrt{(2l-3)(2l+5)/((l-1)(l+2))}$, does not vanish. Consequently $P_l=P(l)$ for every $l\in\Lambda$, and the chain $\Lambda=\{l_0,l_0+2,\dots,l_1\}$ carries a representation if and only if $P(l_0-2)=0$ and $P(l_1)=0$, and the remaining relations are then satisfied identically.
\item The irreducible representations on which $T'$ is hermitian, i.e.\ those with $P_l>0$ along the chain, are: for every charge, the physical levels $\Lambda=\{|q|\text{ or }|q|+1,\dots,k\}$, on which all the matrix elements of the system are recovered, and, for integer $|q|\ge2$, the chains $\Lambda=\{0\text{ or }1,\dots,|q|-2\text{ or }|q|-1\}$ carried by the angular momenta that the monopole removes from the oscillator level. At $|q|=1$ and even $k$ the second family reduces to the singlet $\Lambda=\{0\}$, on which $T'=0$. The parity of $\Lambda$ is part of the definition of the class and not a consequence of the relations: the small chains of the other parity satisfy them as well. For half-integer $q$ there are no representations of the second kind. The chains $\{k+2,k+4,\dots\}$ are infinite and non-unitary, $P(l)<0$ for $l>k$.
\end{enumerate}
\end{theorem}

For integer $q$, the algebra does not know that $l\ge|q|$: when the second family is present, it cuts the symmetric level of $\su(3)$ into two irreducible pieces, the physical one and the small one. The condition $P(l_0-2)=0$ at the bottom of a chain is not vacuous: at $l_0=\tfrac32$ with $|q|\ge\tfrac52$ the rank-3 and rank-2 relations fail on a single angular momentum, and at $l_0=\tfrac12$ the transversality \eqref{eq:transvp} fails, which is why the small representations exist only for integer charges. Restricted to a sector $(k,m)$, these representations give chains on which $(T'_{zz})_{l,l+2}(T'_{zz})_{l+2,l}$ is given by \eqref{eq:5pairs} and the diagonal by \eqref{eq:dl}: \eqref{eq:5pairs} satisfies identically in $l$ the first-order recurrence that the diagonal of \eqref{eq:BC} imposes on the off-diagonal products, and the chains run between consecutive zeros of \eqref{eq:5pairs}, $l_0\in\{M_\pm,M_\pm+1\}$.

\subsection{The cut and the Newman--Penrose operator}\label{sec:eth}
For integer $q$, the second family has a geometric meaning at the level of its carrier spaces: the symmetry algebra detects precisely the kernel of the operator that takes the ordinary harmonics to the monopole harmonics, and this operator is a power of the operator $\eth$ of Newman and Penrose. We recall the spin-weighted spherical harmonics \cite{NP,GMNRS}, fixing the conventions. Throughout, $D^l_{m'm}(\alpha,\beta,\gamma)=\langle l,m'|e^{-i\alpha J_z}e^{-i\beta J_y}e^{-i\gamma J_z}|l,m\rangle=e^{-im'\alpha}d^l_{m'm}(\beta)e^{-im\gamma}$ are the Wigner functions, with $d^l_{m'm}(\beta)=\langle l,m'|e^{-i\beta J_y}|l,m\rangle$ real, and the spherical harmonics carry the Condon--Shortley phase, $Y_{lm}(\theta,\varphi)=\sqrt{(2l+1)/4\pi}\;e^{im\varphi}d^l_{m0}(\theta)$.

\emph{Spin weight.} Let $(e_\theta,e_\varphi)$ be the tangent frame of the unit sphere and $e_\pm=e_\theta\pm ie_\varphi$. A quantity $\eta$ on the sphere has spin weight $s$ if it transforms as $\eta\to e^{is\psi}\eta$ when the frame is rotated by $\psi$ about the normal, $e_\pm\to e^{\pm i\psi}e_\pm$: the component of a tensor field along $p$ factors $e_+$ and $p'$ factors $e_-$ has spin weight $p-p'$, so that a scalar has $s=0$, the components of a vector field along $e_+$ and $e_-$ have $s=\pm1$, and the Stokes combinations $Q\pm iU$ of polarised radiation have $s=\pm2$. The spin weight is a helicity with respect to the local frame. For integer $s$, quantities of spin weight $s$ are sections of the corresponding tensor power of the complex tangent line bundle (or its dual), with absolute Chern number $2|s|$. Half-integer weights require the spin lift and a square-root line bundle. Line bundles over the sphere are classified by this integer, and a section of the monopole bundle of charge $q$, of Chern number $2q$, is therefore the same object as a quantity of spin weight $\pm q$. The gauge phase is the frame phase, and the Dirac condition $2q\in\mathbb Z$ is the condition $2s\in\mathbb Z$. The sign, and the gauge, are fixed below.

\emph{The operators $\eth$ and $\bar\eth$.} In the coordinates $(\theta,\varphi)$,
\begin{equation}\label{eq:eth}
\eth\,\eta=-(\sin\theta)^{s}\Big[\partial_\theta+\frac{i}{\sin\theta}\partial_\varphi\Big]\big((\sin\theta)^{-s}\eta\big),\qquad
\bar\eth\,\eta=-(\sin\theta)^{-s}\Big[\partial_\theta-\frac{i}{\sin\theta}\partial_\varphi\Big]\big((\sin\theta)^{s}\eta\big)
\end{equation}
raise and lower the spin weight by one unit. Writing $\nabla=d-i\mathbf A_s$ with
\begin{equation}\label{eq:As}
\mathbf A_s=-s\cos\theta\,d\varphi ,
\end{equation}
one has $\eth=-(\nabla_\theta+\tfrac{i}{\sin\theta}\nabla_\varphi)$ and $\bar\eth=-(\nabla_\theta-\tfrac{i}{\sin\theta}\nabla_\varphi)$: they are the components of the covariant derivative of a section along the complex null directions, the conjugation by $(\sin\theta)^{\pm s}$ in \eqref{eq:eth} producing the connection term $\mp s\cot\theta$ inside the brackets. The form \eqref{eq:As} is the connection of the weight-$s$ bundle in the frame $(e_\theta,e_\varphi)$. For $s=1$ it is the potential of a monopole of unit charge in the gauge singular at both poles.

\emph{The harmonics ${}_sY_{lm}$.} The spin-weighted harmonics are the quantities of spin weight $s$ that carry the representation $l$ of the rotation group. For every $s\in\tfrac12\mathbb Z$, $l\ge|s|$ and $l-s\in\mathbb Z$, with $m=-l,-l+1,\ldots,l$, we define
\begin{equation}\label{eq:sYD}
{}_sY_{lm}(\theta,\varphi)=e^{i\pi s}\sqrt{\frac{2l+1}{4\pi}}\;e^{im\varphi}\,d^l_{m,-s}(\theta)=e^{i\pi s}\sqrt{\frac{2l+1}{4\pi}}\;\overline{D^l_{m,-s}(\varphi,\theta,0)} ,
\end{equation}
so that ${}_0Y_{lm}=Y_{lm}$. The third Euler angle, absent here, would carry the frame phase $e^{is\psi}$, and half-integer $s$ corresponds to functions on $SU(2)$ rather than $SO(3)$, i.e.\ to half-integer $l$, exactly as half-integer charges do. For fixed $s$ they are orthonormal on the sphere and complete in the space of spin weight $s$. For integer $s$ they are obtained from the ordinary harmonics by raising or lowering the weight,
\begin{equation}\label{eq:sYdef}
{}_sY_{lm}=\sqrt{\frac{(l-s)!}{(l+s)!}}\;\eth^{\,s}\,Y_{lm}\quad(0\le s\le l),\qquad
{}_sY_{lm}=(-1)^s\sqrt{\frac{(l+s)!}{(l-s)!}}\;\bar\eth^{\,-s}\,Y_{lm}\quad(-l\le s\le0),
\end{equation}
and for every $s$,
\begin{gather}\label{eq:ethaction}
\eth\;{}_sY_{lm}=\sqrt{(l-s)(l+s+1)}\;{}_{s+1}Y_{lm},\qquad
\bar\eth\;{}_sY_{lm}=-\sqrt{(l+s)(l-s+1)}\;{}_{s-1}Y_{lm},\\
\bar\eth\eth\;{}_sY_{lm}=-\big(l(l+1)-s(s+1)\big)\,{}_sY_{lm} .\nonumber
\end{gather}
Since $\eth$ annihilates the multiplet $l=s$ and $\bar\eth$ the multiplet $l=-s$, the harmonics ${}_sY_{lm}$ exist only for $l\ge|s|$. In operator form, with $\J^2_{(s)}=-\nabla^2+s^2$ the Casimir operator acting on sections of weight $s$ ($\nabla^2$ the covariant Laplacian of \eqref{eq:As}), \eqref{eq:ethaction} reads $\bar\eth\eth=-(\J^2_{(s)}-s(s+1))$ and $\eth\bar\eth=-(\J^2_{(s)}-s(s-1))$, so the two factorisations intertwine the Casimir operators, $\eth\,\J^2_{(s)}=\J^2_{(s+1)}\eth$.

\emph{Monopole harmonics.} In the gauge of \cite{WuYang} regular away from the negative $z$-axis, $e\mathbf A=q(1-\cos\theta)\,d\varphi=\mathbf A_q+q\,d\varphi$, and the monopole harmonics of \cite{WuYang}, normalised with the phase convention of Edmonds for the ladder operators of $J_z$, are
\begin{equation}\label{eq:WYD}
Y_{q,l,m}(\theta,\varphi)=\sqrt{\frac{2l+1}{4\pi}}\;e^{i(m+q)\varphi}\,d^l_{m,-q}(\theta)=e^{-i\pi q}\,e^{iq\varphi}\;{}_qY_{lm}(\theta,\varphi) .
\end{equation}
Thus, with the conventions \eqref{eq:eth} and \eqref{eq:sYD}, the monopole harmonics of charge $q$ are the spin-weighted harmonics of weight $s=q$, up to the gauge transformation $e^{iq\varphi}$ that takes \eqref{eq:As} to the Wu--Yang potential and a constant phase. The identification goes back to \cite{Dray}, and the cut $l\ge|q|$ of the spectrum is the condition $l\ge|s|$. Since $(\rr\times\pp)^2=\J^2-q^2$, the angular part of $H$ is $\J^2/2r^2$ for every charge and the radial equation does not depend on $q$: the eigenfunctions of $H$ are
\begin{equation}\label{eq:eigenH}
\psi_{nlm}=\phi_{nl}(r)\,Y_{q,l,m}(\theta,\varphi)=e^{-i\pi q}e^{iq\varphi}\,\phi_{nl}(r)\;{}_qY_{lm}(\theta,\varphi),\qquad l\ge|q|,
\end{equation}
with $\phi_{nl}$ the radial functions of the ordinary oscillator. By \eqref{eq:sYdef} and \eqref{eq:ethaction}, the operator $\eth^{q}$ for $q>0$ ($\bar\eth^{|q|}$ for $q<0$), which acts on the angular variables only, satisfies
\begin{equation}\label{eq:intertwine}
\eth^{q}\big(\phi_{nl}Y_{lm}\big)=\sqrt{\frac{(l+q)!}{(l-q)!}}\;\phi_{nl}\;{}_{q}Y_{lm}\quad(l\ge q),\qquad \eth^{q}\big(\phi_{nl}Y_{lm}\big)=0\quad(l<q),
\end{equation}
the multiplet $l=j<q$ being annihilated at the $(j+1)$-th step. It therefore intertwines the oscillator Hamiltonian with $H$ level by level, $\eth^{q}H_0=H\,\eth^{q}$, maps the states of the level $k$ with $l\ge|q|$ isomorphically onto the level $k$ of the charged system, and has for kernel on the level $k$ the states with $l<|q|$: exactly the small piece of Theorem~\ref{thm:reps}. The removal of the states $l<|q|$ is thus a chain of intertwining transformations of the factorisation type, in the angular variable, whose seeds are the lowest multiplets. Its spinorial version, in which the charge changes by half a unit at a time, is the supersymmetry of \cite{DHV_CMP,DHV_PRL}. This is the classical factorisation chain of the spherical top rather than one producing exceptional polynomials: the seeds are ground states and no gap opens inside the sequence of angular momenta.

\section{Recurrence operators from the conserved tensor}\label{sec:hmzcompare}

The recurrence operators of Hoque, Marquette and Zhang \cite{HMZ2018} can be obtained, on each physical sector of fixed energy and magnetic quantum number, by rescaling the parts of $T'_{zz}$ that change the angular momentum by two units. We give this correspondence explicitly and compare the resulting deformed-oscillator structure functions. We use $\omega=1$, $Q=q$ and $\mathbf L=\J$ in the notation of \cite{HMZ2018}.

Let
\[
K=\sqrt{\J^2+\tfrac14},
\qquad
K\,|n,l,m\rangle=\Bigl(l+\tfrac12\Bigr)|n,l,m\rangle,
\]
and let $d(K)$ denote the diagonal part of $T'_{zz}$, so that $d(K)|n,l,m\rangle=d_l|n,l,m\rangle$, where $d_l$ is the diagonal coefficient of Section~\ref{sec:sectors}. For $l=\tfrac12$ its value is zero by the rank-2 selection rule. Define
\[
R_{\pm}=\frac12\Bigl(T'_{zz}-d(K)\pm \frac12[K,T'_{zz}]\Bigr).
\]
Since $R_-=(R_+)^\dagger$ and $[K,R_\pm]=\pm2R_\pm$, the operators $R_+$ and $R_-$ map $(n,l)$ to $(n-1,l+2)$ and $(n+1,l-2)$, respectively. With the actions (2.18)--(2.19) of \cite{HMZ2018} on the basis states, the recurrence generators are, once the phases of the bases are matched,
\[
D_2=R_+\,g(K),
\qquad
D_1=g(K)\,R_-,
\qquad
g(z)=\frac18 z(z+1)(z+2)(H+z-1).
\]
The coefficients of $D_1$, $D_2$ printed in \cite{HMZ2018} have the sign opposite to the positive off-diagonal convention used here; on a sector, the rephasing $|l_0+2j\rangle\mapsto(-1)^j|l_0+2j\rangle$ removes the difference. On a state with $K$-eigenvalue $z$, the raising amplitude is multiplied by $g(z)$ and the lowering amplitude by $g(z-2)$. In particular, the nonzero raising matrix elements obey
\[
\frac{\langle n-1,l+2,m|D_2|n,l,m\rangle}{\langle n-1,l+2,m|R_+|n,l,m\rangle}=g(l+\tfrac12),
\]
which follows from (2.19) of \cite{HMZ2018} and the matrix elements in Section~\ref{sec:sectors}. Consequently,
\[
D_1D_2=\mathcal G(K),
\qquad
D_2D_1=\mathcal G(K-2),
\qquad
\mathcal G(z)=g(z)^2\mathcal F(z).
\]
On the physical sector, $z=l+\tfrac12\ge\tfrac12$ and $H\ge z+1$, so $g(z)>0$. The normalisation therefore preserves every nonzero transition and is invertible on this finite sector.
Here $R_-R_+=\mathcal F(K)$ and \eqref{eq:5pairs} gives
\[
\mathcal F(z)=\frac{\big(H^2-(z+1)^2\big)
\displaystyle\prod_{a=q,m}
\big((z+\tfrac12)^2-a^2\big)\big((z+\tfrac32)^2-a^2\big)}
{z(z+1)^2(z+2)},
\]
with one factor for $a=q$ and one for $a=m$, even when they coincide. The denominator of $\mathcal F$ cancels against $g^2$, and $\mathcal G$ is a polynomial. The expression $\mathcal G(K-2)$ agrees with (2.23) of \cite{HMZ2018}, and on the eigenvectors of the number operator of \cite{HMZ2018}, on which $K=2x+u_{\mathrm{HMZ}}$, its structure function (2.26) is recovered as $\Phi_{\mathrm{HMZ}}(x)=\mathcal G(2x+u_{\mathrm{HMZ}}-2)$, with $u_{\mathrm{HMZ}}=l_0+\tfrac12$ on the physical chain. The comparison is made sector by sector on the matrices; since $K$ involves a square root, the correspondence is not a polynomial change of generators.

The recurrence construction isolates the sector algebra in a form adapted to finite-dimensional representations and the algebraic derivation of the spectrum. The formulas above identify its generators, up to normalisation and basis phase, with the $l\mapsto l\pm2$ parts of $T'_{zz}$, thereby connecting the sector description with the full $\so(3)$-covariant algebra generated by $\J$ and $T'$.

\section{Cubic Casimir and deformed-oscillator representations}\label{sec:deformed}

The sector relations of Theorem~\ref{thm:HR} also admit the deformed-oscillator construction developed for general cubic associative algebras in \cite{Marquette}. This gives a direct algebraic link between Section~\ref{sec:heun}, the representation theory of Section~\ref{sec:reps}, and the recurrence realisation of Section~\ref{sec:hmzcompare}. In the notation of \cite{Marquette}, the cubic algebra parameters are
\begin{gather*}
\alpha=0,\qquad \beta=8,\qquad \gamma=a_5,\qquad \delta=-12,\qquad \epsilon=a_8,\\
\mu=-128,\qquad \nu=a_3,\qquad \xi=a_6,\qquad \zeta=a_9.
\end{gather*}
To avoid confusion with the spectral operator $K=\sqrt{\J^2+\tfrac14}$ of Section~\ref{sec:hmzcompare}, denote the cubic Casimir by $\mathcal C_{\rm cub}$. Specialising the general Casimir of \cite{Marquette} gives
\begin{align}
\mathcal C_{\rm cub}={}&C^2-8\{A,B^2\}-a_5\{A,B\}+76B^2+(8a_5-2a_8)B-64A^4 \notag\\
&+\frac23(a_3-1024)A^3
+\left(a_6+\frac{8a_3+6400}{3}\right)A^2
+(2a_9-4a_3-2048)A .
\label{eq:cubicCasimir}
\end{align}
Direct reduction of $[\mathcal C_{\rm cub},A]$ and $[\mathcal C_{\rm cub},B]$ using only \eqref{eq:AC}--\eqref{eq:BC} and $BA=AB-C$ gives zero. Commutation with $C=[A,B]$ then follows from the Jacobi identity. On the sector representations obtained by restricting the covariant modules of Theorem~\ref{thm:reps}, its value is the scalar
\begin{align}
\kappa(H,m,q)=\frac{16}{9}\Big\{&H^2\big[45m^4+72m^4q^2+48m^2q^4+48m^2q^2-168m^2
+32q^4-160q^2+128\big] \notag\\
&+36(m^2-1)(q^2-1)\big(m^2q^2+2m^2+2q^2-8\big)\Big\}.
\label{eq:cubicCasimirValue}
\end{align}
It is obtained by substituting the matrix elements of the sector, and it does not depend on $l$. The cubic relations alone do not prescribe this value. Like the parameters, it is not symmetric in $(m^2,q^2)$; for the shifted generator $B+\sigma\mathbf 1$, $\sigma=\tfrac23Hq^2$, whose relations are symmetric, the Casimir \eqref{eq:cubicCasimir} with the transformed parameters equals $\mathcal C_{\rm cub}-76\sigma^2+(8a_5-2a_8)\sigma$, and its value is the symmetric expression
\begin{align*}
\frac{16}{9}\Big\{&H^2\big[45(m^4+q^4)+72m^2q^2(m^2+q^2)+42m^2q^2-168(m^2+q^2)+128\big]\\
&+36(m^2-1)(q^2-1)\big(m^2q^2+2m^2+2q^2-8\big)\Big\}.
\end{align*}

We now write the deformed-oscillator realisation of \cite{Marquette} on the restriction to a sector of fixed $m$ of one of the representations of Theorem~\ref{thm:reps}, $l_0$ being the smallest angular momentum of the representation with $l\ge|m|$. Write $l=l_0+2N$, $N=0,1,\ldots,p$, where $l_0+2p$ is the largest angular momentum of the representation, and set
\[
z_0=l_0+\tfrac12,\qquad \mathcal N|N\rangle=N|N\rangle .
\]
For $\beta=8$ and $\delta=-12$, the solution of \cite{Marquette} for $\beta\neq0$ is $A(\mathcal N)=4(\mathcal N+u)^2-\tfrac14$, and the step of two units in $l$ gives $u=z_0/2$, half the $u_{\mathrm{HMZ}}$ of Section~\ref{sec:hmzcompare}. Hence
\begin{equation}
A(\mathcal N)=\big(z_0+2\mathcal N\big)^2-\frac14,
\qquad
b_0(\mathcal N)=-\frac{a_5A(\mathcal N)+a_8}{16A(\mathcal N)-12}.
\label{eq:oscillatorAB}
\end{equation}
For $A(\mathcal N)\neq\tfrac34$, this gives $A(\mathcal N)=l(l+1)$ and $b_0(\mathcal N)=d_l$, with $d_l$ the diagonal entry \eqref{eq:dl}. At $l=\tfrac12$ the denominator of \eqref{eq:oscillatorAB} vanishes; this bottom occurs only for $q^2=\tfrac14$, where $a_5=a_8=0$ and the equation reduces to $0=0$, and the value $b_0=0$ is supplied by the covariant representation, in which the block $(\tfrac12,\tfrac12)$ of $T'$ is absent. After the phase choice used in Section~\ref{sec:hmzcompare}, put
\[
B=b_0(\mathcal N)+R_++R_- .
\]
Then
\[
[\mathcal N,R_\pm]=\pm R_\pm,
\qquad R_+R_-=\Phi_{\mathrm{osc}}(\mathcal N),
\qquad R_-R_+=\Phi_{\mathrm{osc}}(\mathcal N+1),
\]
and the structure function is simply the off-diagonal product of the tensor sector,
\begin{equation}
\Phi_{\mathrm{osc}}(N+1)=\mathcal F(z_0+2N),
\qquad N=0,1,\ldots,p,
\label{eq:sectorStructure}
\end{equation}
where $\mathcal F$ is the factorised function displayed in Section~\ref{sec:hmzcompare}. On a small representation of integer charge, these products are given by the reduced elements of Theorem~\ref{thm:reps}. With the phase convention fixed above, $R_+|N\rangle=\sqrt{\Phi_{\mathrm{osc}}(N+1)}|N+1\rangle$ and $R_-|N\rangle=\sqrt{\Phi_{\mathrm{osc}}(N)}|N-1\rangle$. Here $\Phi_{\mathrm{osc}}(0)=0$ since the lowering operator annihilates the first state; at a missing edge the rational expression is evaluated after cancellation of the removable factors. This is the $\varrho=1$ hermitian gauge of the general realisation $B=b_0(\mathcal N)+R_+\varrho(\mathcal N)+\varrho(\mathcal N)R_-$ of \cite{Marquette}, with $R_+$ and $R_-$ denoting its creation and annihilation operators. The recurrence normalisation of \cite{HMZ2018} replaces $R_\pm$ by $D_2,D_1$ and multiplies the structure function by the square of the spectral factor $g$. The polynomial structure function of the recurrence construction is therefore a rescaled version of the natural deformed-oscillator structure function of the cubic algebra.

A finite-dimensional unitary module obeys the usual conditions
\begin{equation}
\Phi_{\mathrm{osc}}(0)=0,
\qquad
\Phi_{\mathrm{osc}}(p+1)=0,
\qquad
\Phi_{\mathrm{osc}}(N)>0\quad (N=1,\ldots,p).
\label{eq:oscillatorFinite}
\end{equation}
For a sector of a physical carrier, the upper zero is the factor $H^2-(l_p+\tfrac32)^2$ with $l_p=k$. For a sector of a small integer-charge carrier, it is instead one of the factors involving $q$. Every such restriction is therefore a deformed-oscillator module of this form. Whether, conversely, every unitary representation of the cubic algebra of a sector with the value \eqref{eq:cubicCasimirValue} of the Casimir extends to the covariant algebra is not settled here: the relations of the latter at $l=0,\tfrac12,1$ are additional conditions, and the classification of Theorem~\ref{thm:reps} is not replaced by this construction.

\section{Proofs}\label{sec:proofs}

The proofs are arranged so that every step can be checked by hand. Proposition~\ref{prop:dyn} is proved by operator identities (Section~\ref{sec:proofdyn}), the action \eqref{eq:action}--\eqref{eq:Baction} is derived from it (Section~\ref{sec:proofrep}), and the closed forms of the reduced matrix elements follow (Section~\ref{sec:closed}). The relations of the symmetry algebra, Proposition~\ref{prop:S} and Theorems~\ref{thm:T} and \ref{thm:Tp}, are then derived from Proposition~\ref{prop:dyn} by reordering (Section~\ref{sec:reorder}). The Heun--Racah relations of the sectors have been derived from them in the proof of Proposition~\ref{prop:status}. On the reduced matrix elements the covariant relations become identities between rational functions of the quantum numbers, once the $6j$ symbols are written in closed form (Section~\ref{sec:proofcov}). This is the form used for the Heun operators of the sectors and for the representations (Sections~\ref{sec:proofsect}--\ref{sec:proofreps}), and it provides an independent check of Section~\ref{sec:reorder}. All identities have also been checked with a computer algebra system.

\subsection{The dynamical algebra}\label{sec:proofdyn}
Proposition~\ref{prop:dyn} is a consequence of the canonical relations $[r_i,\pi_j]=i\delta_{ij}$, $[\pi_i,\pi_j]=iq\,\epsilon_{ijk}\,r_k/r^3$, and of the vector character of $\rr$, $\pp$, $\J$, i.e.\ $[J_i,V_j]=i\epsilon_{ijk}V_k$ for $V=\rr,\pp,\J$ and $[J_i,S]=0$ for the scalars built from them. Set
\begin{equation}\label{eq:defs2}
\mathbf c=\rr-i\pp,\quad \mathbf n=\frac{\rr}{r^3},\quad K=\tfrac12r^2,\quad P=\tfrac12\pp^2+\frac{q^2}{2r^2},\quad D=\tfrac12(\rr\cdot\pp+\pp\cdot\rr),\quad p_r=\tfrac12(\hat\rr\cdot\pp+\pp\cdot\hat\rr),
\end{equation}
so that $H=K+P$, $D=\rr\cdot\pp-\tfrac32 i$ and $p_r=\hat\rr\cdot\pp-i/r$. Since $[H,r^2]=-2iD$ and $\mathbf c\cdot\mathbf c=r^2-2iD-\pp^2$, and since $\tfrac12\epsilon_{ijk}\{J_j,c_k\}=\epsilon_{ijk}J_jc_k-\tfrac12\epsilon_{ijk}[J_j,c_k]$, the generators \eqref{eq:bB} read
\begin{equation}\label{eq:bc}
b_i^\dagger=(\J\times\mathbf c)_i-i\,c_i,\qquad b_i=(\J\times\mathbf c^\dagger)_i-i\,c_i^\dagger,\qquad
B_\mp=\tfrac12\Big(\mathbf c^{(\dagger)}\!\cdot\mathbf c^{(\dagger)}-\frac{q^2}{r^2}\Big),\qquad B_\pm=\mathcal K-P\pm iD .
\end{equation}

\begin{lemma}\label{lem:basic}
\begin{enumerate}[label=(\roman*)]
\item $[c_i,c_j]=[c_i^\dagger,c_j^\dagger]=-iq\,\epsilon_{ijk}n_k$ and $[c_i,c_j^\dagger]=-2\delta_{ij}+iq\,\epsilon_{ijk}n_k$.
\item $\J\times\J=i\J$, $\J\cdot\rr=\rr\cdot\J=-qr$, $\J\cdot\mathbf n=-q/r^2$, $\J\cdot\mathbf c=-q\,(r-ip_r)$ and $\J\cdot\mathbf c^\dagger=-q\,(r+ip_r)$.
\item $\mathbf c\cdot\mathbf c^\dagger=2H-3-q^2/r^2$, $\mathbf c^\dagger\!\cdot\mathbf c=2H+3-q^2/r^2$, $\mathbf c\times\mathbf c^\dagger=2i\J+iq(2+r^{-2})\hat\rr$, $\mathbf c\times\mathbf c=\mathbf c^\dagger\times\mathbf c^\dagger=-iq\,\mathbf n$.
\item $[P,\pi_j]=\dfrac{iq}{r^3}J_j$, $[\mathcal K,\pi_j]=ir_j$, $[P,r_j]=-i\pi_j$, $[D,r_j]=-ir_j$ and $[D,\pi_j]=i\pi_j$.
\item $[\mathcal K,D]=2iK$, $[P,D]=-2iP$, $[\mathcal K,P]=iD$, and $2(KP+PK)=D^2+\J^2-\tfrac34$.
\item $\pp^2=p_r^2+(\J^2-q^2)/r^2$ and $\J^2=r^2\pp^2-(\rr\cdot\pp)^2+i\,\rr\cdot\pp+q^2$.
\end{enumerate}
\end{lemma}

\begin{proof}
(i) is immediate from the canonical relations. (ii) $\J\times\J=i\J$ is $[J_i,J_j]=i\epsilon_{ijk}J_k$. Since $\epsilon_{ijk}r_j\pi_kr_i=\epsilon_{ijk}r_jr_i\pi_k-i\epsilon_{ijk}\delta_{ik}r_j=0$, $(\rr\times\pp)\cdot\rr=\rr\cdot(\rr\times\pp)=0$ and $\J\cdot\rr=\rr\cdot\J=-q\,\hat\rr\cdot\rr=-qr$, and $1/r^3$ commutes with $\J$, whence $\J\cdot\mathbf n$. Next, $(\pp\times\pp)_j=\tfrac12\epsilon_{jab}[\pi_a,\pi_b]=iq\,n_j$, so $(\rr\times\pp)\cdot\pp=\rr\cdot(\pp\times\pp)=iq/r$ and $\J\cdot\pp=iq/r-q\,\hat\rr\cdot\pp$. Hence $\J\cdot\mathbf c=-qr-i\J\cdot\pp=-q(r-r^{-1}-i\,\hat\rr\cdot\pp)=-q(r-ip_r)$, and $\J\cdot\mathbf c^\dagger$ likewise. (iii) $\mathbf c\cdot\mathbf c^\dagger=r^2+\pp^2+i(\rr\cdot\pp-\pp\cdot\rr)=r^2+\pp^2-3$, $\mathbf c\times\mathbf c^\dagger=i\,\rr\times\pp-i\,\pp\times\rr+\pp\times\pp=2i(\J+q\hat\rr)+iq\,\mathbf n$, because $\pp\times\rr=-\rr\times\pp$ ($\epsilon_{jab}[\pi_a,r_b]=0$), and $\mathbf c\times\mathbf c=-\pp\times\pp$. (iv) $[\tfrac12\pp^2,\pi_j]=\tfrac12(\pi_i[\pi_i,\pi_j]+[\pi_i,\pi_j]\pi_i)=iq\,\epsilon_{ijk}n_k\pi_i$, since $\epsilon_{ijk}[\pi_i,n_k]=-i\epsilon_{ijk}\partial_in_k=0$ ($\mathbf n$ is a gradient). This is $(iq/r^3)(\rr\times\pp)_j$, and $[q^2/2r^2,\pi_j]=-iq^2r_j/r^4=-(iq/r^3)\,q\hat r_j$. The other relations are $[f(\rr),\pi_j]=i\partial_jf$ and the fact that $D$ generates dilations, $[D,f]=-i\,\rr\cdot\nabla f$, $[D,\pi_j]=i\pi_j$ (here $\rr\cdot[\pp,\pi_j]=iq\,\epsilon_{ijk}r_ir_k/r^3=0$ is used). (v) The first three follow from (iv), $D$ being of degree $0$ and $\mathcal K$, $P$ of degrees $2$, $-2$ under dilations, and $[\mathcal K,P]=\tfrac14[r^2,\pp^2]=iD$. For the last one, $\pp^2r^2=r^2\pp^2-4i\,\rr\cdot\pp-6$ gives $2(KP+PK)=r^2\pp^2-2i\,\rr\cdot\pp-3+q^2$, while $D^2=(\rr\cdot\pp)^2-3i\,\rr\cdot\pp-\tfrac94$, and the claim follows from the expression of $\J^2$ in (vi). (vi) $(\rr\times\pp)^2=r_j\pi_kr_j\pi_k-r_j\pi_kr_k\pi_j=r^2\pp^2-(\rr\cdot\pp)^2+i\,\rr\cdot\pp$, the term $\epsilon_{kjm}r_jr_kn_m$ produced by $[\pi_k,\pi_j]$ vanishing. Then $(\rr\times\pp)\cdot\hat\rr=\hat\rr\cdot(\rr\times\pp)=0$ gives $\J^2=(\rr\times\pp)^2+q^2$, and $p_r=r^{-1}(\rr\cdot\pp-i)$ gives $p_r^2=r^{-2}\big((\rr\cdot\pp)^2-i\,\rr\cdot\pp\big)$.
\end{proof}

\begin{proof}[Proof of Proposition~\ref{prop:dyn}]
\emph{The $\so(2,1)$ relations \eqref{eq:dyn1}.} With $B_\pm=\mathcal K-P\pm iD$ and Lemma~\ref{lem:basic}(v), $[H,B_\pm]=[\mathcal K+P,\mathcal K-P\pm iD]=-2iD\pm i(2iK-2iP)=\mp2B_\pm$, $[B_+,B_-]=-2i[\mathcal K-P,D]=4(\mathcal K+P)$, and
\[
B_-B_+=(\mathcal K-P)^2+i[\mathcal K-P,D]+D^2=(\mathcal K+P)^2-2(KP+PK)-2H+D^2=H^2-2H+\tfrac34-\J^2=\Cc .
\]

\emph{The ladder relations \eqref{eq:dyn2} and $[H,b]=-b$.} By Lemma~\ref{lem:basic}(iv),
\begin{equation}\label{eq:anom}
[H,c_j]=c_j+\frac{q}{r^3}J_j,\qquad [B_+,c_j^\dagger]=\frac{q}{r^3}J_j,\qquad [B_-,c_j^\dagger]=-2c_j+\frac{q}{r^3}J_j,\qquad [B_-,c_j]=-\frac{q}{r^3}J_j ,
\end{equation}
and the adjoint relations: the centrifugal term makes the anomalous parts of these commutators proportional to $\J$. Now the map $V\mapsto\J\times V-iV$, which sends $\mathbf c$ to $b^\dagger$ and $\mathbf c^\dagger$ to $b$, annihilates $\J$ by Lemma~\ref{lem:basic}(ii), and commutes with the scalars $H$, $B_\pm$, $r^{-3}$. Applying it to \eqref{eq:anom} gives $[H,b^\dagger]=b^\dagger$, $[B_+,b]=0$, $[B_-,b]=-2b^\dagger$, $[B_-,b^\dagger]=0$, and the adjoints. This is the mechanism by which the term $q^2/2r^2$ of McIntosh and Cisneros restores the ladder structure. Finally $\J\cdot b=\epsilon_{ijk}J_iJ_jc_k^\dagger-i\J\cdot\mathbf c^\dagger=\tfrac12\epsilon_{ijk}[J_i,J_j]c^\dagger_k-i\J\cdot\mathbf c^\dagger=0$, and $b\cdot\J=\epsilon_{ijk}J_jc_k^\dagger J_i-i\mathbf c^\dagger\!\cdot\J=(-i+2i)\J\cdot\mathbf c^\dagger-i\,\mathbf c^\dagger\!\cdot\J=0$, using $\epsilon_{ijk}[c^\dagger_k,J_i]=2ic^\dagger_j$ and $\mathbf c^\dagger\!\cdot\J=\J\cdot\mathbf c^\dagger$.

\emph{The scalar relations \eqref{eq:dyn4}.} For vector operators $U$, $V$ one has, by $\epsilon_{iab}\epsilon_{icd}=\delta_{ac}\delta_{bd}-\delta_{ad}\delta_{bc}$ and the vector character,
\begin{gather}\label{eq:vecid}
(\J\times U)\cdot(\J\times V)=\J^2\,U\cdot V-(\J\cdot U)(\J\cdot V)+i\,\J\cdot(U\times V),\\
(\J\times U)\cdot V=\J\cdot(U\times V),\qquad U\cdot(\J\times V)=-\J\cdot(U\times V)+2i\,U\cdot V .\nonumber
\end{gather}
Hence $b\cdot b=(\J^2+1)\,\mathbf c^\dagger\!\cdot\mathbf c^\dagger-(\J\cdot\mathbf c^\dagger)^2+i\J\cdot(\mathbf c^\dagger\times\mathbf c^\dagger)$ and $b^\dagger\!\cdot b=\J^2\,\mathbf c\cdot\mathbf c^\dagger-(\J\cdot\mathbf c)(\J\cdot\mathbf c^\dagger)+i\J\cdot(\mathbf c\times\mathbf c^\dagger)+\mathbf c\cdot\mathbf c^\dagger$. By Lemma~\ref{lem:basic}(ii),(iii), $i\J\cdot(\mathbf c^\dagger\times\mathbf c^\dagger)=-q^2/r^2$, $i\J\cdot(\mathbf c\times\mathbf c^\dagger)=-2\J^2+2q^2+q^2/r^2$, $(\J\cdot\mathbf c^\dagger)^2=q^2(r+ip_r)^2$ and $(\J\cdot\mathbf c)(\J\cdot\mathbf c^\dagger)=q^2(r-ip_r)(r+ip_r)$. By (vi) and $[r,p_r]=i$,
\[
(r+ip_r)^2=\mathbf c^\dagger\!\cdot\mathbf c^\dagger+\frac{\J^2-q^2}{r^2},\qquad (r-ip_r)(r+ip_r)=r^2+p_r^2-1=2H-1-\frac{\J^2}{r^2} .
\]
Substituting, $b\cdot b=(\J^2+1-q^2)(\mathbf c^\dagger\!\cdot\mathbf c^\dagger-q^2/r^2)=2\Aa B_+$ and $b^\dagger\!\cdot b=(2H-3)(\J^2+1-q^2)-2\J^2$. Applying $[B_-,\,\cdot\,]$ to $b\cdot b=2\Aa B_+$ gives moreover $b^\dagger\!\cdot b+b\cdot b^\dagger=4\Aa H$, i.e.\ $b\cdot b^\dagger=\Aa(2H+3)+2\J^2$.

\emph{The tensor relations \eqref{eq:dyn3}.} A vector operator $V$ is determined by its component $V_+=V_x+iV_y$, since $V_z=-\tfrac12[J_-,V_+]$ and $V_-=[J_-,V_z]$, and a symmetric traceless tensor operator $A_{ij}$ is determined by $A_{++}=A_{xx}-A_{yy}+2iA_{xy}$ in the same way. We use $[J_z,V_\pm]=\pm V_\pm$, $[J_\pm,V_z]=\mp V_\pm$, $[J_\pm,V_\mp]=\pm2V_z$, $[J_\pm,V_\pm]=0$, and Lemma~\ref{lem:basic}(i) in the form $[c^\dagger_z,c^\dagger_\pm]=\mp q\,n_\pm$, $[c^\dagger_+,c^\dagger_-]=-2q\,n_z$, $[c^\dagger_z,c_z]=2$, $[c^\dagger_+,c_+]=0$, $[c^\dagger_z,c_+]=-[c^\dagger_+,c_z]=q\,n_+$. From $(\J\times V)_+=-iJ_+V_z+iJ_zV_+$ and $(\J\times V)_z=-\tfrac{i}{2}(J_-V_+-J_+V_-)$,
\begin{equation}\label{eq:bcomp}
b_+=-i\big[J_+c^\dagger_z-(J_z-1)c^\dagger_+\big],\qquad b^\dagger_+=-i\big[J_+c_z-(J_z-1)c_+\big],\qquad b_z=-\tfrac{i}{2}\big(J_-c^\dagger_+-J_+c^\dagger_-\big)-ic^\dagger_z .
\end{equation}
The commutator $[b_i,b_j]$ is antisymmetric, hence the vector $V_k=\tfrac12\epsilon_{kij}[b_i,b_j]$, with $V_+=i[b_z,b_+]$, and the first relation of \eqref{eq:dyn3} is $[b_z,b_+]=-2J_+B_+$. Expanding with \eqref{eq:bcomp}, $[b_z,b_+]=-\tfrac12\mathcal U-\mathcal W$ with
\begin{align*}
\mathcal U&=[J_-c^\dagger_+-J_+c^\dagger_-,\;J_+c^\dagger_z-(J_z-1)c^\dagger_+]\\
&=q\big[J_-J_+n_++J_+^2n_-+2J_+(J_z-1)n_z\big]-2c^\dagger_zc^\dagger_++2J_+\big(c^{\dagger2}_z+c^\dagger_-c^\dagger_+\big),\\
\mathcal W&=[c^\dagger_z,\;J_+c^\dagger_z-(J_z-1)c^\dagger_+]=c^\dagger_+c^\dagger_z+q(J_z-1)n_+ ,
\end{align*}
so that, with $c^\dagger_zc^\dagger_+-c^\dagger_+c^\dagger_z=-qn_+$,
\[
[b_z,b_+]=-J_+\big(c^{\dagger2}_z+c^\dagger_-c^\dagger_+\big)-qJ_zn_+-\tfrac q2\big(J_-J_+n_++J_+^2n_-+2J_+J_zn_z\big).
\]
Since $c^\dagger_-c^\dagger_+=\tfrac12\{c^\dagger_+,c^\dagger_-\}+qn_z$ and $\mathbf c^\dagger\!\cdot\mathbf c^\dagger=c^{\dagger2}_z+\tfrac12\{c^\dagger_+,c^\dagger_-\}$, the first term is $-2J_+B_+-q^2J_+/r^2-qJ_+n_z$. Moreover, $J_-J_+n_++J_+^2n_-+2J_+J_zn_z=J_+(J_-n_++J_+n_-+2J_zn_z)-2J_zn_+=2J_+\,\J\cdot\mathbf n-2J_zn_+=-2qJ_+/r^2-2J_zn_+$. The remaining terms cancel and $[b_z,b_+]=-2J_+B_+$.

For $M_{ij}=[b_i,b_j^\dagger]$ we decompose into ranks. The trace is $b\cdot b^\dagger-b^\dagger\!\cdot b=6\Aa+4\J^2$, which is the trace of the right side of \eqref{eq:dyn3}. The antisymmetric part follows from the first relation. Applying $[B_-,\,\cdot\,]$ to $[b_i,b_j]=-2i\epsilon_{ijk}J_kB_+$ with $[B_-,b]=-2b^\dagger$ and $[B_-,B_+]=-4H$ gives $M_{ij}-M_{ji}=-4i\epsilon_{ijk}J_kH$. For the symmetric traceless part it suffices to compute $M_{++}=[b_+,b^\dagger_+]$. With \eqref{eq:bcomp},
\begin{gather*}
[J_+c^\dagger_z,J_+c_z]=J_+c^\dagger_+c_z+2J_+^2-J_+c_+c^\dagger_z,\qquad
[J_+c^\dagger_z,(J_z-1)c_+]=qJ_+(J_z-1)n_+-J_+c_+c^\dagger_z,\\
[(J_z-1)c^\dagger_+,J_+c_z]=-q(J_z-1)J_+n_++J_+c_zc^\dagger_+,\qquad
[(J_z-1)c^\dagger_+,(J_z-1)c_+]=0 ,
\end{gather*}
whence
\[
[b_+,b^\dagger_+]=-\big\{2J_+^2+J_+[c^\dagger_+,c_z]+q[J_z,J_+]n_+\big\}=-\big\{2J_+^2-qJ_+n_++qJ_+n_+\big\}=-2J_+^2 ,
\]
which is the $++$ component of $-\{J_i,J_j\}+\tfrac23\delta_{ij}\J^2$. The three parts assemble into the second relation of \eqref{eq:dyn3}.
\end{proof}

\subsection{From the dynamical algebra to the representation}\label{sec:proofrep}
By the Wigner--Eckart theorem, $\langle n',l',m'|b^\dagger_\mu|n,l,m\rangle=\langle l\,m\,1\,\mu|l'\,m'\rangle\,\rho_{(n',l'),(n,l)}$ with $\rho$ independent of $m$, and the reduced elements of $b$ follow by conjugation,
\begin{equation}\label{eq:rhob}
\rho^{b}_{(n',l'),(n,l)}=(-1)^{l-l'}\sqrt{\frac{2l+1}{2l'+1}}\;\rho^{b^\dagger}_{(n,l),(n',l')} .
\end{equation}
The reduced element of a product $[V\otimes W]^{(L)}$ of two tensor operators of ranks $j_a$ and $j_b$ between $|l\rangle$ and $|l'\rangle$ is obtained from the recoupling identity
\begin{multline}\label{eq:recoupling}
\sum_{\mu\nu}\langle j_a\mu\,j_b\nu|LM\rangle\langle l\,m\,j_b\,\nu|l''\,m+\nu\rangle\langle l''\,m+\nu\,j_a\,\mu|l'm'\rangle\\
=(-1)^{L+l+l'}\sqrt{(2l''+1)(2L+1)}\sixj{l}{j_b}{l''}{j_a}{l'}{L}\langle l\,m\,L\,M|l'm'\rangle ,
\end{multline}
as a sum over the intermediate $l''$ of $6j$ symbols times products of reduced elements.

The $6j$ symbols and Clebsch--Gordan coefficients needed below have one entry equal to $1$ or $2$ and are elementary, and their closed forms are collected in Appendix~\ref{app:6j}.

\begin{proposition}\label{prop:rep}
Let $V$ be a representation of the relations of Proposition~\ref{prop:dyn}, with $b_i=(b_i^\dagger)^\dagger$, $B_+=(B_-)^\dagger$, generated by the vectors $|0,l_0,m\rangle$ of an $\so(3)$ multiplet annihilated by $b$ and $B_+$, and in which the joint eigenspaces of $(H,\J^2,J_z)$ are one-dimensional. Assume in addition that $q\in\tfrac12\mathbb Z$, $l_0-q\in\mathbb Z$, and that $b^\dagger$ does not annihilate the ground multiplet. Then $l_0=|q|$, $H|0,l_0,m\rangle=(l_0+\tfrac32)|0,l_0,m\rangle$, and the action of $b^\dagger$ and $B_+$ is given by \eqref{eq:action}--\eqref{eq:Baction}, up to the phases of the basis vectors.
\end{proposition}

\begin{proof}
On the ground multiplet, where $H$ takes the value $h_0$, $B_-B_+=\Cc$ gives $(h_0-\tfrac12)(h_0-\tfrac32)=l_0(l_0+1)$ and $b^\dagger\!\cdot b=\Aa(2H-3)-2\J^2$ gives $(l_0(l_0+1)+1-q^2)(2h_0-3)=2l_0(l_0+1)$, whence $h_0=l_0+\tfrac32$ for the $\so(2,1)$ representation to be of the positive discrete series and for $b^\dagger$ not to annihilate the ground multiplet, and then $l_0(l_0+1)=l_0(l_0(l_0+1)+1-q^2)$, i.e.\ $l_0=|q|$ or $l_0=0$. The branch $l_0=0$ with $q\ne0$ gives $\|b_z^\dagger|0,0,0\rangle\|^2=2(1-q^2)\le0$ and is excluded by the hypothesis on $b^\dagger$ (at $|q|=1$, $b^\dagger$ annihilates the ground multiplet and the representation is trivial). The relations \eqref{eq:dyn1}, restricted to the $\so(2,1)$ generated by $B_\pm$ and $H$, fix \eqref{eq:Baction} in each sector $l$.

In the real phase convention of \eqref{eq:action}, let $k=2n+l$ and assume the reduced elements of $b^\dagger$ leaving the levels $k'<k$ known. For the state $(n,l)$ of level $k$ the unknowns are $X=\rho_{(n,l+1),(n,l)}^2$ and $Y=\rho_{(n+1,l-1),(n,l)}^2$, and the known quantities are $X'=\rho^2_{(n,l),(n,l-1)}$ and $Y'=\rho^2_{(n,l),(n-1,l+1)}$, the squares of the elements arriving at $(n,l)$ from the level $k-1$. Since $\sum_{m',\mu}\langle l\,m\,1\,\mu|l'm'\rangle^2=(2l'+1)/(2l+1)$ and $\sum_{m',\mu}\langle l'\,m'\,1\,\mu|l\,m\rangle^2=1$, the scalars $b\cdot b^\dagger$ and $b^\dagger\!\cdot b$ act on $(n,l)$ as $\big((2l+3)X+(2l-1)Y\big)/(2l+1)$ and $X'+Y'$. The trace of the second relation \eqref{eq:dyn3} gives
\begin{equation}\label{eq:sys1}
(2l+3)X+(2l-1)Y=(2l+1)\big(10\lambda_l+6-6q^2+X'+Y'\big),\qquad \lambda_l=l(l+1),
\end{equation}
and its symmetric traceless part, $[b\otimes b^\dagger]^{(2)}-[b^\dagger\otimes b]^{(2)}=-2J^{(2)}$, evaluated on $(n,l)$ by \eqref{eq:recoupling} and \eqref{eq:rhob} with the symbols $\sixj{l}{1}{l\pm1}{1}{l}{2}$ of Appendix~\ref{app:6j}, gives
\begin{equation}\label{eq:sys2}
(2l-1)(2l+3)\big(lX+(l+1)Y\big)=(2l+1)\big[2\lambda_l(2l-1)(2l+3)+(l+1)(2l+3)X'+l(2l-1)Y'\big].
\end{equation}
The determinant of the system \eqref{eq:sys1}--\eqref{eq:sys2} in $(X,Y)$ is $3(2l+1)(2l-1)(2l+3)$, nonzero for $l\ne\tfrac12$, so that $X$ and $Y$ are determined by $X'$ and $Y'$, and one checks by substitution that
\begin{equation}\label{eq:XY}
X=\frac{2(2n+2l+3)\big((l+1)^2-q^2\big)(l+1)}{2l+3},\qquad Y=\frac{4(n+1)(l^2-q^2)\,l}{2l-1},
\end{equation}
with $X'=X(n,l-1)$ and $Y'=Y(n-1,l+1)$, solves it. At $l=|q|$ it gives $Y=0$, consistently with the absence of the state $(n+1,|q|-1)$. At $l=\tfrac12$ on the physical branch the downward state is absent, so $Y=0$ and the scalar equation alone fixes $X$. This is \eqref{eq:action}, with absent transitions set to zero before evaluation. The relation $b\cdot b=2\Aa B_+$ then fixes the sign of $\rho_{(n,l+1),(n,l)}\rho_{(n+1,l-1),(n,l)}$ relative to that of \eqref{eq:Baction}, the products of amplitudes around closed paths such as $(n,l)\to(n,l+1)\to(n+1,l)\to(n,l)$ being independent of the phases of the basis. With the phases fixed by $c_0,c_1\ge0$, the sign is read off $[B_-,b_z]=-2b_z^\dagger$: on the transition $|0,l,m\rangle\to|0,l+1,m\rangle$ it gives $-\beta_1c_1(0,l+1,m)=-2c_0(0,l,m)$, hence $\beta_1=2\sqrt{l+\tfrac32}>0$, which is the positive sign in \eqref{eq:Baction}, and the remaining signs follow along the ladder.\end{proof}

The reduced elements of $b^\dagger$ are thus
\begin{gather}\label{eq:rho}
\rho_{(n,l+1),(n,l)}=\sqrt{\frac{2(2n+2l+3)(l+1-q)(l+1+q)(l+1)}{2l+3}},\\
\rho_{(n+1,l-1),(n,l)}=-\sqrt{\frac{4(n+1)(l-q)(l+q)\,l}{2l-1}} .\notag
\end{gather}

\subsection{Closed forms}\label{sec:closed}
For a rank-2 tensor operator $A$ commuting with $H$, with reduced elements $A_{l'l}$ vanishing unless $l'\in\{l,l\pm2\}$, we write
\begin{equation}\label{eq:alphaP}
\alpha_l(A)=\frac{A_{ll}}{j_2(l)},\qquad \Pc_l(A)=\sqrt{(2l+1)(2l+5)}\;A_{l+2,l}A_{l,l+2},
\end{equation}
so that $\alpha_l$ is the diagonal in units of the quadrupole of $\J$, and $\Pc_l$ is the product of the two off-diagonal elements, which is independent of the phases of the basis. Since $T_{ij}=P_{ij}+P_{ji}-\tfrac23\delta_{ij}P_{kk}$ and the symmetric traceless part of $P_{ij}=b_i^\dagger b_j$ is the rank-2 part of the product $b^\dagger_ib_j$, the tensor of Section~\ref{sec:algebra} is $T=2\,[b^\dagger\otimes b]^{(2)}$ in the normalisation of the Conventions, and likewise $S=2\,[b^\dagger\otimes b^\dagger]^{(2)}B_+$. On the level $k$, the block $(l+2,l)$ of $T$ is reached from $(n,l)$ through $(n-1,l+1)$ only, and the block $(l,l)$ through $(n,l-1)$ and $(n-1,l+1)$. By \eqref{eq:recoupling}, \eqref{eq:rhob}, $\sixj{l}{1}{l+1}{1}{l+2}{2}^2=1/[5(2l+3)]$ and the symbols $\sixj{l}{1}{l\pm1}{1}{l}{2}$, with $X$, $Y$ as in \eqref{eq:XY},
\begin{gather}
\Pc_l:=\Pc_l(T)=\Pc_l(T')=\frac{4(2l+1)(2l+5)}{2l+3}\,X(n-1,l+1)\,Y(n-1,l+1)\notag\\
\phantom{\Pc_l:=\Pc_l(T)=\Pc_l(T')}=\frac{16\,\langle q\rangle_l\langle-q\rangle_l\,(k-l)(k+l+3)\,(l+1)(l+2)}{2l+3},\label{eq:Pcl}\\
\alpha_l(T)=-\frac{2\big[(l+1)(2l+3)X(n,l-1)+l(2l-1)Y(n-1,l+1)\big]}{\lambda_l(2l-1)(2l+3)}=\alpha'_l-2(H-1),\label{eq:alphal}\\
\alpha'_l:=\alpha_l(T')=\frac{2(4q^2-1)H}{(2l-1)(2l+3)},\notag
\end{gather}
where $\langle c\rangle_l=(l+1+c)(l+2-c)$, so that $\langle q\rangle_l\langle-q\rangle_l=((l+1)^2-q^2)((l+2)^2-q^2)$, and the last equality in \eqref{eq:alphal} is the identity $(2l+3)(2n+2l+1)(l^2-q^2)+2n(2l-1)((l+1)^2-q^2)=\tfrac12(2l+1)\big[(H-1)(2l-1)(2l+3)-(4q^2-1)H\big]$ with $H=2n+l+\tfrac32$. The matrix elements of $T'_{zz}$ and $T_{zz}$ follow from $T_{zz}=\sqrt{2/3}\,T_0$ and the coefficients $\langle l\,m\,2\,0|l'\,m\rangle$ of the Appendix: $(T'_{zz})_{ll}=\tfrac13(3m^2-\lambda_l)\alpha'_l$ is \eqref{eq:dl}, and $(T_{zz})_{l,l+2}(T_{zz})_{l+2,l}=\tfrac23\langle l\,m\,2\,0|l+2\,m\rangle\langle l+2\,m\,2\,0|l\,m\rangle\,\Pc_l/\sqrt{(2l+1)(2l+5)}$ is \eqref{eq:5pairs}. Relations \eqref{eq:PQred} and Proposition~\ref{prop:S} are derived from the dynamical algebra in Section~\ref{sec:reorder}, and on the reduced elements they can also be read off directly. The antisymmetric part $P_{ij}-P_{ji}=\epsilon_{ijk}(b^\dagger\times b)_k$ is the vector $b^\dagger\times b=-i\sqrt2\,[b^\dagger\otimes b]^{(1)}$, whose reduced element on $(n,l)$, computed with $\sixj{l}{1}{l\pm1}{1}{l}{1}$, is $\tfrac{i}{j_1(l)}\big[(l+1)X(n,l-1)-l\,Y(n-1,l+1)\big]$, and $b^\dagger\times b=i\J\Fi$ is the identity
\begin{equation}\label{eq:PQanti}
(l+1)\,X(n,l-1)-l\,Y(n-1,l+1)=\lambda_l\,\Fi_l,\qquad \Fi_l=2(\lambda_l-q^2-k),
\end{equation}
which holds by \eqref{eq:XY}. Similarly, the blocks $(l+2,l)$ of $S$ and $T$ are reached through the same intermediate state $(n-1,l+1)$ and differ by the replacement of the amplitude $\rho^b_{(n-1,l+1),(n,l)}=-\sqrt{(2l+1)/(2l+3)}\;\rho_{(n,l),(n-1,l+1)}$ by $\rho_{(n-1,l+1),(n-1,l)}\,\beta_n$, where $\beta_n=2\sqrt{n(n+l+\tfrac12)}$ is the amplitude \eqref{eq:Baction}. By \eqref{eq:rho} the ratio is $2n+2l+1=k+l+1$, and similarly $S_{l,l+2}=(k-l-2)T_{l,l+2}$ ; these are the blocks $(l+2,l)$ and $(l,l+2)$ of $(H-2)T+\tfrac14[\J^2,T]$, since $\lambda_{l+2}-\lambda_l=4(l+\tfrac32)$.

\subsection{The symmetry algebra from the dynamical algebra}\label{sec:reorder}
Proposition~\ref{prop:dyn} gives the commutator of every pair of generators among $\J$, $H$, $B_\pm$, $b^\dagger$, $b$, and therefore allows any product of generators to be brought to a normal order, say $\J$, $H$, $B_-$, $B_+$, $b^\dagger$, $b$. The relations of Section~\ref{sec:algebra} are obtained by computing suitable elements in two ways. We carry this out here, so that Proposition~\ref{prop:S} and Theorems~\ref{thm:T} and \ref{thm:Tp} appear as consequences of Proposition~\ref{prop:dyn} alone, without reference to the basis $|n,l,m\rangle$.

\begin{theorem}\label{thm:reorder}
Relations \eqref{eq:PQred}, Proposition~\ref{prop:S} and Theorems~\ref{thm:T} and \ref{thm:Tp} hold in every associative algebra with involution generated by an $\so(3)$ triple $\J$, scalars $H$, $B_\pm=B_\mp^\dagger$ and vector operators $b$, $b^\dagger$ satisfying \eqref{eq:dyn1}--\eqref{eq:dyn4}.
\end{theorem}

We use the components $V_\pm=V_x\pm iV_y$ of Section~\ref{sec:proofdyn}, with $[J_z,V_\pm]=\pm V_\pm$, $[J_\pm,V_z]=\mp V_\pm$, $[J_\pm,V_\mp]=\pm2V_z$, and write $P_{ab}=b_a^\dagger b_b$ for $a,b\in\{+,z,-\}$. The components of \eqref{eq:dyn3} and of its adjoint that occur are
\begin{gather*}
[b_z,b_+]=-2J_+B_+,\qquad [b^\dagger_+,b^\dagger_z]=2B_-J_+,\qquad M_{++}=[b_+,b^\dagger_+]=-2J_+^2,\\
M_{z+}=[b_z,b_+^\dagger]=-\{J_z,J_+\}-2HJ_+,\qquad M_{+z}=[b_+,b_z^\dagger]=-\{J_z,J_+\}+2HJ_+ .
\end{gather*}
The adjoints of \eqref{eq:dyn2} are $[B_-,b^\dagger_i]=0$, $[B_+,b^\dagger_i]=2b_i$. In the conventions of Section~\ref{sec:algebra} the highest components of the tensors are $T_2=P_{++}$, $T_1=-(P_{+z}+P_{z+})$, $S_2=b^{\dagger2}_+B_+$, $S^\dagger_2=B_-b_+^2$, $J^{(2)}_2=\tfrac12J_+^2$, $J^{(3)}_3=-\dfrac{J_+^3}{2\sqrt2}$, $[J\otimes T]^{(3)}_3=-J_+T_2/\sqrt2$, and $[J\otimes V]^{(2)}_2=\tfrac12J_+V_+$ for a vector $V$. Two tensor operators of the same rank are equal as soon as their highest components are, the others being obtained by commutation with $J_-$, and this is used throughout. The identities of $\so(3)$ that are needed are collected first.

\begin{lemma}\label{lem:so3}
Let $A$, $B$, $T$ be rank-2 tensor operators, $V$ a vector operator, and $\J\cdot T=J_iT_{ij}$, which equals $T\cdot\J$.
\begin{enumerate}[label=(\roman*)]
\item $\epsilon_{mik}A_{il}B_{lk}=i\sqrt{\tfrac52}\,[A\otimes B]^{(1)}_m$, the symmetric traceless part of $A_{il}B_{lk}+A_{kl}B_{li}$ is $-\sqrt{\tfrac73}\,[A\otimes B]^{(2)}$, and $A_{il}B_{li}=\sqrt5\,[A\otimes B]^{(0)}$.
\item $[\J^2,V_i]=-2i(\J\times V)_i-2V_i$ and $[\J^2,T_{ij}]=-2i\big((\J\times T)_{ij}+(\J\times T)_{ji}\big)-6T_{ij}$, where $(\J\times T)_{ij}=\epsilon_{iab}J_aT_{bj}$. Moreover
\begin{equation}\label{eq:HRgen}
[\J^2,[\J^2,T]]=8\{\J^2,T\}-12\,T-24\,[J\otimes(\J\cdot T)]^{(2)} .
\end{equation}
\item Besides \eqref{eq:shiftid},
\begin{gather}
[T\otimes J^{(2)}]^{(1)}+[J^{(2)}\otimes T]^{(1)}=3\sqrt{\tfrac25}\;\J\cdot T,\qquad [J^{(2)}\otimes J^{(2)}]^{(1)}=\sqrt{\tfrac25}\,(\J^2-\tfrac34)\,\J,\notag\\
[T\otimes J^{(2)}]^{(2)}+[J^{(2)}\otimes T]^{(2)}=\tfrac{1}{\sqrt{21}}\Big(2\{\J^2,T\}+9T-12\,[J\otimes(\J\cdot T)]^{(2)}\Big),\label{eq:shiftid2}\\
[J^{(2)}\otimes J^{(2)}]^{(2)}=-\tfrac{2}{\sqrt{21}}(\J^2-\tfrac{15}{4})\,J^{(2)},\notag\\
T_{ij}J_{(ij)}=J_{(ij)}T_{ij}=\J\cdot(\J\cdot T),\qquad J_{(ij)}J_{(ij)}=\tfrac23\J^4-\tfrac12\J^2,\qquad \J\cdot J^{(2)}=(\tfrac23\J^2-\tfrac12)\,\J .\notag
\end{gather}
\end{enumerate}
\end{lemma}

\begin{proof}
(i) The three contractions are covariant and bilinear, with $A$ to the left of $B$, and each rank occurs once in the product of two rank-2 tensors, so the constants are read off one component, e.g.\ $(A_{il}B_{lk}-A_{kl}B_{li})$ for $(i,k)=(x,y)$ against $\langle2\mu\,2\nu|1\,0\rangle$. (ii) For a tensor operator of rank $k$, $\sum_j[J_j,[J_j,X]]=k(k+1)X$, so that $[\J^2,X]=\sum_j(J_j[J_j,X]+[J_j,X]J_j)=2\sum_jJ_j[J_j,X]-k(k+1)X$. With $[J_j,V_i]=i\epsilon_{jik}V_k$ this is the first formula, and the second follows from $[J_a,T_{ij}]=i\epsilon_{aik}T_{kj}+i\epsilon_{ajk}T_{ik}$. Both sides of \eqref{eq:HRgen} are rank-2 tensors, and on the component $\mu=2$ one has $[\J^2,T_2]=2T_1J_++4T_2J_z+6T_2$, $[\J^2,T_1]=2J_zT_1+\sqrt6J_+T_0+2J_-T_2-6T_1$, and $(\J\cdot T)_+=-\tfrac{1}{\sqrt6}J_+T_0+J_-T_2-J_zT_1$, so that both sides become, once the $J$'s are moved to the right of the $T$'s and $J_\pm J_\mp$ is expressed through $\J^2$ and $J_z$, $12T_1J_+J_z+28T_1J_++2\sqrt6\,T_0J_+^2+4T_2\J^2+12T_2J_z^2+44T_2J_z+36T_2$. (iii) Each identity relates tensors of the same rank and is checked on the highest component in the same way, with $J^{(2)}_1=-\tfrac12\{J_+,J_z\}$, $J^{(2)}_0=\tfrac{1}{\sqrt6}(3J_z^2-\J^2)$ and the coupling coefficients of the Appendix. For instance $[T\otimes J^{(2)}]^{(3)}_3+[J^{(2)}\otimes T]^{(3)}_3=\tfrac{1}{\sqrt2}\big([T_2,J^{(2)}_1]+[J^{(2)}_2,T_1]\big)=\tfrac{1}{\sqrt2}(2J_+T_2+2J_+T_2)=-4[J\otimes T]^{(3)}_3$, and $[J^{(2)}\otimes J^{(2)}]^{(3)}_3=\tfrac{1}{\sqrt2}[J^{(2)}_2,J^{(2)}_1]=J_+^3/\sqrt2=-2J^{(3)}_3$, which are \eqref{eq:shiftid}. The scalar identities are contractions: $T_{ij}J_{(ij)}=T_{ij}J_iJ_j$ since $T$ is symmetric and traceless, and $J_iJ_{(ij)}=(\tfrac23\J^2-\tfrac12)J_j$ by $J_iJ_jJ_i=\J^2J_j-J_j$.
\end{proof}

\paragraph{Antisymmetric part and traces.}
The traces in \eqref{eq:PQred} are \eqref{eq:dyn4} and its adjoint: $P_{kk}=b^\dagger\!\cdot b$, and $Q_{kk}=(b\cdot b)^\dagger B_+=2B_-\Aa B_+=2\Aa\,\Cc$. For the antisymmetric part, $(b^\dagger\times b)_j=\epsilon_{jki}P_{ki}=-iJ_iP_{ji}$ by the covariance computation of Section~\ref{sec:algebra}, and $J_iP_{ji}=J_i(b_ib^\dagger_j-M_{ij})=-J_iM_{ij}$ by $\J\cdot b=0$. From \eqref{eq:dyn3}, $J_iM_{ij}=2(2\J^2+1-q^2)J_j-J_i\{J_i,J_j\}-2i\epsilon_{ijk}J_iJ_kH$, with $J_i\{J_i,J_j\}=2\J^2J_j-J_j$ and $\epsilon_{ijk}J_iJ_k=-iJ_j$ by $\J\times\J=i\J$. Hence $J_iM_{ij}=\Fi J_j$ and $b^\dagger\times b=i\J\Fi$, which is the first relation of \eqref{eq:PQred}. We shall use it in the forms
\begin{equation}\label{eq:R1}
\epsilon_{ijm}P_{ij}=iJ_m\Fi,\qquad J_kP_{mk}=P_{km}J_k=-J_m\Fi,\qquad P_{z+}=P_{+z}+J_+\Fi ,
\end{equation}
the second following from $\J\cdot b=b^\dagger\!\cdot\J=0$ and the vector character of $b$, $b^\dagger$.

\paragraph{Proposition~\ref{prop:S}.}
Since $[\J^2,T_2]=2T_1J_++4T_2J_z+6T_2$ and $P_{z+}=P_{+z}+J_+\Fi$, the components $\mu=2$ of the two relations of Proposition~\ref{prop:S} read
\begin{equation}\label{eq:Spp}
B_-b_+^2=(H-\tfrac72)\,T_2-T_2J_z+P_{z+}J_+ ,\qquad b_+^{\dagger2}B_+=(H-\tfrac12)\,T_2+T_2J_z-P_{+z}J_+ ,
\end{equation}
and it suffices to prove these. For the first, compute $b^\dagger_+\Cc\,b_+$ in two ways: by $B_-B_+=\Cc$, $[B_-,b^\dagger_+]=0$ and $[B_+,b^\dagger_+]=2b_+$,
\[
b^\dagger_+\Cc\,b_+=b^\dagger_+B_-B_+b_+=B_-(B_+b^\dagger_+-2b_+)b_+=\Cc\,b^\dagger_+b_+-2B_-b_+^2,\qquad\text{i.e.}\qquad 2B_-b_+^2=[\Cc,b^\dagger_+]\,b_+ .
\]
Now $[H,b^\dagger_+]=b^\dagger_+$ gives $[(H-\tfrac12)(H-\tfrac32),b^\dagger_+]=(2H-3)b^\dagger_+$, and Lemma~\ref{lem:so3}(ii) gives $[\J^2,b^\dagger_+]=-2b^\dagger_zJ_++2b^\dagger_+(J_z+1)$, so that $[\Cc,b^\dagger_+]=2(H-\tfrac52)b^\dagger_++2b^\dagger_zJ_+-2b^\dagger_+J_z$. Moving $J_+$ and $J_z$ to the right of $b_+$ gives the first relation of \eqref{eq:Spp}. For the second, let $\tilde\Cc=B_+B_-=\Cc+4H=(H+\tfrac12)(H+\tfrac32)-\J^2$, and by $[B_-,b_+]=-2b^\dagger_+$ and $[B_+,b_+]=0$,
\[
b_+\tilde\Cc\,b^\dagger_+=b_+B_+B_-b^\dagger_+=B_+(B_-b_++2b^\dagger_+)b^\dagger_+=\tilde\Cc\,b_+b^\dagger_++2B_+b^{\dagger2}_+,\qquad\text{i.e.}\qquad 2B_+b^{\dagger2}_+=[b_+,\tilde\Cc]\,b^\dagger_+ ,
\]
with $[b_+,\tilde\Cc]=(2H+5)b_+-2b_zJ_++2b_+J_z$. Then $b^{\dagger2}_+B_+=B_+b^{\dagger2}_+-2(b_+b^\dagger_++b^\dagger_+b_+)$, and with $b_+b^\dagger_+=T_2-2J_+^2$, $b_zb^\dagger_+=P_{+z}+M_{z+}$ and the $J$'s moved to the right, the terms in $J_+^2$ cancel and the second relation of \eqref{eq:Spp} follows. This proves Proposition~\ref{prop:S}. Adding the two relations \eqref{eq:Spp} gives $S_2+S^\dagger_2=2(H-2)T_2+\Fi J_+^2$, which is used below.

\paragraph{The rank-3 relation.}
The highest component of \eqref{eq:TT3} is $[T_2,T_1]=-8(H-1)J_+T_2-4\Fi J_+^3$. Write $-[T_2,T_1]=[P_{++},P_{+z}]+[P_{++},P_{z+}]$ and expand with $[XY,ZW]=X[Y,Z]W+XZ[Y,W]+[X,Z]WY+Z[X,W]Y$, which is \eqref{eq:PPcomm} in components:
\begin{align*}
[P_{++},P_{+z}]&=b^\dagger_+M_{++}b_z+b^{\dagger2}_+[b_+,b_z]-b^\dagger_+M_{z+}b_+=-2b^\dagger_+J_+^2b_z+2b^{\dagger2}_+J_+B_+-b^\dagger_+M_{z+}b_+,\\
[P_{++},P_{z+}]&=b^\dagger_+M_{+z}b_++[b^\dagger_+,b^\dagger_z]b_+^2-b^\dagger_zM_{++}b_+=b^\dagger_+M_{+z}b_++2B_-J_+b_+^2+2b^\dagger_zJ_+^2b_+ .
\end{align*}
Since $M_{+z}-M_{z+}=4HJ_+$ and $b^\dagger_+H=(H-1)b^\dagger_+$, the terms in $M$ add up to $4(H-1)J_+T_2$. Moving the $J$'s to the left, $b^\dagger_+J_+^2=J_+^2b^\dagger_+$ and $b^\dagger_zJ_+^2=J_+^2b^\dagger_z+2J_+b^\dagger_+$, the terms in $J_+^2$ give $2J_+^2(P_{z+}-P_{+z})+4J_+T_2=2\Fi J_+^3+4J_+T_2$, and the terms in $B_\pm$ give $2J_+(S_2+S^\dagger_2)=4(H-2)J_+T_2+2\Fi J_+^3$. Altogether $-[T_2,T_1]=8(H-1)J_+T_2+4\Fi J_+^3$, as claimed. This is the mechanism announced after \eqref{eq:PPcomm}: the term $8(H-1)[J\otimes T]^{(3)}$ collects $4(H-1)$ from the $M$'s, $4$ from the reordering of $J_+^2$ and $4(H-2)$ from Proposition~\ref{prop:S}, and $8\Fi J^{(3)}$ comes from the antisymmetric part \eqref{eq:R1} and from Proposition~\ref{prop:S} in equal parts.

\paragraph{The relations of rank 1, 2 and 0.}
These are contractions, and we derive them in Cartesian form. Let $F=P_{kk}=b^\dagger\!\cdot b$ and $G=b\cdot b^\dagger=F+M_{kk}=\Aa(2H+3)+2\J^2$, and by \eqref{eq:dyn4}, $F=(2H-5)\J^2+(2H-3)(1-q^2)=2\Cc+(H-\tfrac32)\Fi$. By \eqref{eq:R1},
\begin{equation}\label{eq:TP}
T_{ij}=2P_{ij}-i\epsilon_{ijk}J_k\Fi-\tfrac23\delta_{ij}F .
\end{equation}
The quartic terms reduce through $G$: $P_{il}P_{lk}=b^\dagger_i\,b_lb^\dagger_l\,b_k=b^\dagger_iGb_k$, and, since $b^\dagger_iH=(H-1)b^\dagger_i$ and $b^\dagger_i\J^2=\J^2b^\dagger_i+2i(\J\times b^\dagger)_i+2b^\dagger_i$ by Lemma~\ref{lem:so3}(ii),
\begin{equation}\label{eq:quartic}
P_{il}P_{lk}=(2H+3)\big[(\J^2+2)P_{ik}+2i(\J\times P)_{ik}\big]+(2H+1)(1-q^2)P_{ik},\qquad (\J\times P)_{ik}=\epsilon_{iab}J_aP_{bk} .
\end{equation}
The cross terms of \eqref{eq:TP} reduce through \eqref{eq:R1} and $\J\times\J=i\J$.

\emph{Trace.} Expanding $T_{ij}T_{ji}$ with \eqref{eq:TP}, the cross terms are $-2i\epsilon_{jim}P_{ij}J_m\Fi=-2\J^2\Fi^2$, then $-2i\epsilon_{ijk}J_k\Fi P_{ji}=-2\J^2\Fi^2$, then $-\epsilon_{ijk}\epsilon_{jim}J_kJ_m\Fi^2=2\J^2\Fi^2$, and finally $-\tfrac83F^2+\tfrac43F^2$, so that
\[
\Tr T^2=4P_{ij}P_{ji}-2\J^2\Fi^2-\tfrac43F^2 .
\] By \eqref{eq:quartic} with $i=k$ summed, using $\epsilon_{iab}J_aP_{bi}=i\J^2\Fi$,
\[
P_{ij}P_{ji}=(2H+3)\big[(\J^2+2)F-2\J^2\Fi\big]+(2H+1)(1-q^2)F ,
\]
and substituting $F$ and $\J^2=\tfrac12\Fi+H-\tfrac32+q^2$ gives \eqref{eq:trT2}.

\emph{Rank 1.} Let $V_m=\epsilon_{mik}T_{il}T_{lk}=i\sqrt{5/2}\,[T\otimes T]^{(1)}_m$. Expanding with \eqref{eq:TP} and contracting the $\epsilon$'s, the eight cross terms are $-2iJ_m\Fi(F+\Fi)$ twice, $i\Fi^2J_m$, and four terms $\pm\tfrac43 iJ_m\Fi F$ of which two cancel two, so that $V_m=4\epsilon_{mik}P_{il}P_{lk}-iJ_m\Fi(4F+3\Fi)$. By \eqref{eq:quartic}, $\epsilon_{mik}P_{ik}=iJ_m\Fi$ and $\epsilon_{mik}(\J\times P)_{ik}=J_kP_{mk}-J_mF=-J_m(\Fi+F)$,
\[
\epsilon_{mik}P_{il}P_{lk}=iJ_m\big[(2H+3)(\J^2\Fi-2F)+(2H+1)(1-q^2)\Fi\big],
\]
hence $V_m=iJ_m\big[4(2H+3)(\J^2\Fi-2F)+4(2H+1)(1-q^2)\Fi-4\Fi F-3\Fi^2\big]$. After substitution of $F$ and $\J^2$ the bracket is $3\chi$, with $\chi$ as in \eqref{eq:TT1}, and $[T\otimes T]^{(1)}=-i\sqrt{2/5}\,V=3\sqrt{2/5}\,\chi\,\J$, which is \eqref{eq:TT1}.

\emph{Rank 2.} Let $U_{ik}=T_{il}T_{lk}+T_{kl}T_{li}$, whose traceless part is $-\sqrt{7/3}\,[T\otimes T]^{(2)}$. The expansion of $T_{il}T_{lk}$ with \eqref{eq:TP} contains, besides $4P_{il}P_{lk}$, the terms $-2i\epsilon_{lkb}P_{il}J_b\Fi$ and $-2i\epsilon_{ila}J_a\Fi P_{lk}$, which are $-2i(\J\times P^{t})_{ki}\Fi-4P_{ik}\Fi-2P_{ki}\Fi+2\delta_{ik}F\Fi$ and $2i\Fi(\J\times P)_{ik}$ with $(\J\times P^t)_{ki}=\epsilon_{kab}J_aP_{ib}$, obtained by moving $J_b$ through $b_l$ and $b^\dagger_i$. It also contains the term $-\epsilon_{ila}\epsilon_{lkb}J_aJ_b\Fi^2=-\Fi^2(J_kJ_i-\delta_{ik}\J^2)$, the terms $-\tfrac43\{P_{ik},F\}+\tfrac49\delta_{ik}F^2$, and terms antisymmetric in $(i,k)$. With $(\J\times P^t)_{ik}=(\J\times P)_{ik}-i\Fi(\delta_{ik}\J^2-J_kJ_i)$, from \eqref{eq:R1}, and $W_{ik}=(\J\times P)_{ik}+(\J\times P)_{ki}$, the symmetrisation gives
\[
U_{ik}=4(P_{il}P_{lk}+P_{kl}P_{li})+2i[\Fi,W_{ik}]+\Fi^2\big(\{J_i,J_k\}-2\delta_{ik}\J^2\big)-6T_{ik}\Fi-\tfrac43\{T_{ik},F\}-\tfrac89\delta_{ik}F^2 .
\]
Now $P$ is $\tfrac12T$ plus its antisymmetric part $\tfrac{i}{2}\epsilon J\Fi$ and its trace, so that by Lemma~\ref{lem:so3}(ii)
\[
W_{ik}=\tfrac{i}{4}\big([\J^2,T_{ik}]+6T_{ik}\big)+\tfrac{i}{2}\Fi\big(2\delta_{ik}\J^2-\{J_i,J_k\}\big),\qquad 2i[\Fi,W_{ik}]=-[\J^2,[\J^2,T_{ik}]]-6[\J^2,T_{ik}] ,
\]
and \eqref{eq:quartic} gives $4(P_{il}P_{lk}+P_{kl}P_{li})=4(2H+3)\big[(\J^2+2)(P_{ik}+P_{ki})+2iW_{ik}\big]+4(2H+1)(1-q^2)(P_{ik}+P_{ki})$ with $P_{ik}+P_{ki}=T_{ik}+\tfrac23\delta_{ik}F$. Taking the traceless part, writing $4(2H+3)\J^2T=2(2H+3)(\{\J^2,T\}+[\J^2,T])$, $-6T\Fi=-3\{T,\Fi\}+6[\J^2,T]$, and using \eqref{eq:HRgen} with \eqref{eq:transv}, $[\J^2,[\J^2,T]]=8\{\J^2,T\}-12T+16(2\Cc+H\Fi)J^{(2)}$, all commutators cancel and
\[
U^{\rm traceless}=\tfrac43(H-1)\{\J^2,T\}+\big(\tfrac{20}{3}H-6-\tfrac83q^2H\big)T+\big(2\Fi^2+24\Fi-32\Cc\big)J^{(2)} ,
\]
where $\{T,\Fi\}=2\{\J^2,T\}+2(3-2H-2q^2)T$ and $\{T,F\}=(2H-5)\{\J^2,T\}+2(2H-3)(1-q^2)T$ were used. This is $-\tfrac43$ times the bracket of \eqref{eq:TT2}, as it should be since $\sqrt{7/3}\cdot4/\sqrt{21}=\tfrac43$.

\paragraph{Theorem~\ref{thm:Tp}.}
With $T'=T+\beta J^{(2)}$, $\beta=2(H-1)$, Lemma~\ref{lem:so3}(iii) turns each relation of Theorem~\ref{thm:T} into the corresponding relation of Theorem~\ref{thm:Tp}. The transversality \eqref{eq:transvp} is \eqref{eq:transv} plus $\beta\,\J\cdot J^{(2)}$, and the rank-3 relation is the computation of the proof of Proposition~\ref{prop:shift}. For the others, with $\J\cdot T=-\tfrac23(2\Cc+H\Fi)\J$ and $[J\otimes(\J\cdot T)]^{(2)}=-\tfrac23(2\Cc+H\Fi)J^{(2)}$,
\begin{align*}
[T'\otimes T']^{(1)}&=[T\otimes T]^{(1)}-2\sqrt{\tfrac25}\,\beta\,(2\Cc+H\Fi)\,\J+\sqrt{\tfrac25}\,\beta^2(\J^2-\tfrac34)\,\J,\\
[T'\otimes T']^{(2)}&=[T\otimes T]^{(2)}+\tfrac{\beta}{\sqrt{21}}\Big(2\{\J^2,T\}+9T+8(2\Cc+H\Fi)J^{(2)}\Big)-\tfrac{2\beta^2}{\sqrt{21}}(\J^2-\tfrac{15}{4})J^{(2)},\\
\Tr T'^2&=\Tr T^2-\tfrac43\,\beta\,(2\Cc+H\Fi)\,\J^2+\beta^2(\tfrac23\J^4-\tfrac12\J^2),
\end{align*}
in which $T=T'-\beta J^{(2)}$ and $\{\J^2,J^{(2)}\}=2\J^2J^{(2)}$ are substituted, and the first is $3f_1=3\chi-4(H-1)(2\Cc+H\Fi)+4(H-1)^2(\J^2-\tfrac34)$. Finally \eqref{eq:HR1} is \eqref{eq:HRgen} with $\J\cdot T'=\tfrac13(4q^2-1)H\J$. In each case the right sides of Theorem~\ref{thm:Tp} are recovered by direct expansion. This completes the proof of Theorem~\ref{thm:reorder}.

\subsection{The covariant relations on the reduced matrix elements}\label{sec:proofcov}
Each relation of Theorems~\ref{thm:T} and \ref{thm:Tp} is an equality between rank-$L$ tensor operators commuting with $H$. By the Wigner--Eckart theorem it is equivalent to the equality of the reduced elements of its two sides on every block $(l',l)$ of every level, and only the blocks $l'=l$ and $l'=l\pm2$ can be nonzero. This form of the relations is what the classification of the representations in Section~\ref{sec:proofreps} uses, and, since the reduced elements of $T$ are known in closed form, it also provides a check of Section~\ref{sec:reorder} independent of the reordering. The reduced elements of the left sides are given by \eqref{eq:recoupling} in terms of \eqref{eq:alphaP}, and with the $6j$ symbols of Appendix~\ref{app:6j} one finds:

\begin{lemma}\label{lem:red}
Let $A$ be as in \eqref{eq:alphaP}, $\alpha_l=\alpha_l(A)$, $\Pc_l=\Pc_l(A)$. The reduced elements of $[A\otimes A]^{(L)}$ on the block $(l,l)$, divided by $j_L(l)$ ($j_0=1$), are
\begin{align}
L=0:&\quad \frac{1}{\sqrt5}\Big[\frac{\Pc_l+\Pc_{l-2}}{2l+1}+\frac{\lambda_l(2l-1)(2l+3)}{6}\,\alpha_l^2\Big],\label{eq:red0}\\
L=1:&\quad \sqrt{\tfrac{2}{5}}\Big[-\frac{\Pc_l}{(l+1)(2l+1)}+\frac{\Pc_{l-2}}{l(2l+1)}+\frac{(2l-1)(2l+3)}{4}\,\alpha_l^2\Big],\label{eq:red1}\\
L=2:&\quad \sqrt{\tfrac{3}{7}}\Big[\frac{2\,\Pc_l}{(l+1)(2l+1)(2l+3)}+\frac{2\,\Pc_{l-2}}{l(2l-1)(2l+1)}-\frac{(2l-3)(2l+5)}{6}\,\alpha_l^2\Big],\label{eq:red2}\\
L=3:&\quad -\frac{\Pc_l}{(l+1)(l+2)(2l+1)(2l+3)}+\frac{\Pc_{l-2}}{l(l-1)(2l-1)(2l+1)}-2\,\alpha_l^2 ,\label{eq:red3}
\end{align}
and on the block $(l+2,l)$, divided by $A_{l+2,l}$,
\begin{equation}\label{eq:redoff}
L=2:\ \ \frac{1}{\sqrt{21}}\big[l(2l-1)\,\alpha_l+(l+3)(2l+7)\,\alpha_{l+2}\big],\qquad
L=3:\ \ \sqrt{\frac{l(l+3)}{12}}\,\big[(2l-1)\,\alpha_l-(2l+7)\,\alpha_{l+2}\big].
\end{equation}
The reduced elements of $\J\cdot A=-\sqrt{5/3}\,[J\otimes A]^{(1)}$ and of $[J\otimes A]^{(3)}$ are, on $(l,l)$ and divided by $j_1(l)$, resp.\ $j_3(l)$, $\tfrac16(2l-1)(2l+3)\,\alpha_l$ and $\alpha_l$. On $(l+2,l)$, $\J\cdot A$ vanishes and $[J\otimes A]^{(3)}$ is $\sqrt{l(l+3)/3}\;A_{l+2,l}$. Finally $\{\J^2,A\}$ has reduced elements $(\lambda_l+\lambda_{l'})A_{l'l}$ and $J^{(2)}$ has $j_2(l)\delta_{l'l}$.
\end{lemma}

The closed forms of Appendix~\ref{app:6j} are those of integer $l$, while for half-integer $l$ every symbol used here changes sign, as does the phase $(-1)^{l+l'}$ of \eqref{eq:recoupling}, so that the lemma holds for all $l$. In terms of $s=l+\tfrac32$ and of the notation
\begin{equation}\label{eq:Ps}
\Pc(s):=\Pc_l=\frac{8\,(s^2-(q-\tfrac12)^2)(s^2-(q+\tfrac12)^2)(H^2-s^2)(s^2-\tfrac14)}{s},\qquad \alpha'(s):=\alpha'_l=\frac{(4q^2-1)H}{2s(s-2)},
\end{equation}
so that $\Pc_{l-2}=\Pc(s-2)$, $\lambda_l=s(s-2)+\tfrac34$, $\Aa_l=s(s-2)+\tfrac74-q^2$, the relations of Theorem~\ref{thm:Tp} become the following identities of rational functions of $(s,H,q)$, which can be checked by direct expansion:
\begin{align}
&\text{\eqref{eq:transvp}:}\quad \tfrac16(2l-1)(2l+3)\,\alpha'_l=\tfrac13(4q^2-1)H ;\notag\\
&\text{\eqref{eq:TT3p}, }(l+2,l):\quad (2l-1)\,\alpha'_l=(2l+7)\,\alpha'_{l+2};\notag\\
&\text{\eqref{eq:TT2p}, }(l+2,l):\quad l(2l-1)\,\alpha'_l+(l+3)(2l+7)\,\alpha'_{l+2}=2(4q^2-1)H;\notag\\
&\text{\eqref{eq:TT3p}, }(l,l):\quad -\frac{\Pc(s)}{s(s-1)(2s-1)(2s+1)}+\frac{\Pc(s-2)}{(s-1)(s-2)(2s-3)(2s-5)}-2\alpha'(s)^2\notag\\
&\hspace{7cm}=8\big(2\Aa_l-H^2\big);\label{eq:id3}\\
&\text{\eqref{eq:TT2p}, }(l,l):\quad 3\Big[\frac{\Pc(s)}{s(s-1)(2s-1)}+\frac{\Pc(s-2)}{(s-1)(s-2)(2s-3)}-\frac{2(s-3)(s+1)}{3}\,\alpha'(s)^2\Big]\notag\\
&\hspace{7cm}=2(4q^2-1)H\,\alpha'(s)+4p_2;\label{eq:id2}\\
&\text{\eqref{eq:TT1p}:}\quad -\frac{\Pc(s)}{(s-1)(2s-1)}+\frac{\Pc(s-2)}{(s-1)(2s-3)}+s(s-2)\,\alpha'(s)^2=3f_1;\label{eq:id1}\\
&\text{\eqref{eq:trTp2}:}\quad \frac{\Pc(s)+\Pc(s-2)}{2(s-1)}+\tfrac23\,\lambda_l\,s(s-2)\,\alpha'(s)^2=\operatorname{tr}T'^2 ,\label{eq:id0}
\end{align}
where $f_1$, $p_2$ and $\operatorname{tr}T'^2$ are the right sides of \eqref{eq:TT1p}, \eqref{eq:TT2p}, \eqref{eq:trTp2} with $\J^2=\lambda_l$. The first three are immediate from \eqref{eq:alphal}. In \eqref{eq:id3} both sides are polynomials of degree $2$ in $H$ once the left side is brought to a common denominator, and similarly for the others, and in each case the two sides agree term by term. Theorem~\ref{thm:T} is obtained in the same way with $\alpha_l=\alpha'_l-2(H-1)$ in place of $\alpha'_l$, the same $\Pc$, and the right sides of Theorem~\ref{thm:T}: the blocks $(l+2,l)$ read $(2l-1)\alpha_l-(2l+7)\alpha_{l+2}=16(H-1)$ and $l(2l-1)\alpha_l+(l+3)(2l+7)\alpha_{l+2}=4\big[(1-H)(\lambda_l+\lambda_{l+2})+\tfrac92-(5-2q^2)H\big]$, the transversality is $\tfrac16(2l-1)(2l+3)\alpha_l=-\tfrac23(2\Cc_l+H\Fi_l)$, and the diagonal blocks are \eqref{eq:id3}--\eqref{eq:id0} with $\alpha'$ replaced by $\alpha$ and the right sides replaced by $8(H-1)\alpha_l+8\Fi_l$, $4\big[2(1-H)\lambda_l\alpha_l+(\tfrac92-(5-2q^2)H)\alpha_l+24\Cc_l-\tfrac32\Fi_l(\Fi_l+12)\big]$, $3\chi$ and $\operatorname{tr}T^2$. All are identities in $(s,H,q)$, as they must be by Theorem~\ref{thm:reorder}. The $6j$ symbols and the reduced elements depend on $q$ only through $q^2$ and polynomially, and on $l$ through rational functions and the compensating signs mentioned above. At $q^2=\tfrac14$, $T'$ is transverse, $\J\cdot T'=0$, and $T'_{zz}$ has no diagonal, $a_5=a_8=0$.

\subsection{Sectors: the Heun operators and the Heun--Racah relations}\label{sec:proofsect}
In the coupled basis, with $\varkappa=k_1+k_2+j$ and $M=k_1+k_2+N$, the Hahn partner $Y=K_0^{(2)}$ has diagonal and off-diagonal products \cite{VdJ}
\begin{gather}\label{eq:hahn}
Y_{jj}=M\,\frac{\varkappa(\varkappa-1)-k_1(k_1-1)+k_2(k_2-1)}{2\varkappa(\varkappa-1)},\\
Y_{j,j+1}Y_{j+1,j}=\frac{(j+1)(j+2k_1+2k_2-1)(j+2k_1)(j+2k_2)(N-j)(N+j+2k_1+2k_2)}{(2\varkappa)^2(2\varkappa-1)(2\varkappa+1)} ,\nonumber
\end{gather}
the first by the projection theorem applied to $K_0^{(1)}=K_0-K_0^{(2)}$. At $j=0$ the first formula reduces to $Y_{00}=Mk_2/(k_1+k_2)$, which is the value to be used when $k_1+k_2=1$. With the identifications of Proposition~\ref{prop:heunhahn}, $2\varkappa=s=l+\tfrac32$, $2M=H$, and the six factors of the products are $(s^2-(l_0-\tfrac12)^2)(s^2-4(k_1-k_2)^2)(H^2-s^2)/64$ over $s^2(s-1)(s+1)$, while \eqref{eq:5pairs} is the product of the five factors $s^2-c_i^2$ over the same denominator. Hence the ratio of the off-diagonal products of $T'_{zz}$ to those of $Y$ is $64(s^2-a^2)(s^2-b^2)$. For a bilinear operator $\tau_1\J^2Y+\tau_2Y\J^2+\tau_3\J^2+\tau_4Y+\tau_0$ this ratio is $(\tau_1\lambda_l+\tau_2\lambda_{l+2}+\tau_4)(\tau_1\lambda_{l+2}+\tau_2\lambda_l+\tau_4)=[(\tau_1+\tau_2)(s^2+\tfrac34)+\tau_4]^2-4(\tau_2-\tau_1)^2s^2$, and the two expressions agree exactly for the values of $\tau_1+\tau_2$, $\tau_2-\tau_1$, $\tau_4$ given in \eqref{eq:heunhahn}. For the diagonal, write $c=(l_0-\tfrac12)\,2|k_1-k_2|$, so that $Y_{ll}=\tfrac H4\big[1+\eta\,c/s(s-2)\big]$ with $\eta=\varepsilon$, by \eqref{eq:hahn}, and $\lambda_l=s(s-2)+\tfrac34$, $(T'_{zz})_{ll}=-(4q^2-1)H\big(s(s-2)+\tfrac34-3m^2\big)/6s(s-2)$ by \eqref{eq:dl}. The difference $(T'_{zz})_{ll}-8\lambda_lY_{ll}+2H\lambda_l-\tau_4Y_{ll}$ is therefore
\[
-\frac{H}{s(s-2)}\Big[\frac{(4q^2-1)\big(s(s-2)+\tfrac34-3m^2\big)}{6}+\eta c\,\big(2s(s-2)+\tfrac32+\tfrac{\tau_4}{4}\big)\Big]-\frac{\tau_4H}{4},
\]
and it is independent of $l$ if and only if the terms of degree zero in $s(s-2)$ inside the bracket cancel, i.e.\ $(4q^2-1)(4m^2-1)/8=\eta c\,(\tfrac32+\tfrac{\tau_4}4)=2\eta\varepsilon\,c\,ab$, using $\tau_4=-6+8\varepsilon ab$. Since the product of the four numbers $l_0-\tfrac12$, $2|k_1-k_2|$, $a$, $b$ is $(4m^2-1)(4q^2-1)/16$, this holds exactly when $\eta=\varepsilon$, and then the constant is $\tau_0=-\tfrac H4\big[\tau_4+\tfrac23(4q^2-1)+8\eta c\big]$. A diagonal $D$ with $D_l/D_{l+2}=(T'_{zz})_{l,l+2}/[(\tau_1\lambda_l+\tau_2\lambda_{l+2}+\tau_4)Y_{l,l+2}]$ completes the proof of Proposition~\ref{prop:heunhahn}.

Theorem~\ref{thm:HR} has been derived from Theorem~\ref{thm:Tp} in the proof of Proposition~\ref{prop:status}. In the sector it can also be read directly on the matrices, which is how the coefficients were first obtained. Let $A=\J^2$, $B=T'_{zz}$ in the sector, with $B_{ll}=d_l$ given by \eqref{eq:dl}, $x_l=B_{l,l+2}B_{l+2,l}$ given by \eqref{eq:5pairs}, and $C=[A,B]$, so that $C_{ll'}=(\lambda_l-\lambda_{l'})B_{ll'}$. Relation \eqref{eq:AC} has entries $(\lambda_l-\lambda_{l'})^2B_{ll'}$ on the left: on the blocks $l'=l\pm2$ it is the identity $(\lambda_{l+2}-\lambda_l)^2=8(\lambda_l+\lambda_{l+2})-12$ of the quadratic grid, and on the diagonal it reads $(16\lambda_l-12)d_l+a_5\lambda_l+a_8=0$, which is \eqref{eq:dl}. Relation \eqref{eq:BC} has entries $\sum_{l''}B_{ll''}B_{l''l'}(2\lambda_{l''}-\lambda_l-\lambda_{l'})$ on the left: on the blocks $l'=l\pm4$ it is $2\lambda_{l+2}-\lambda_l-\lambda_{l+4}=-8$, which fixes the coefficient of $B^2$. On the blocks $l'=l\pm2$ it is $(\lambda_{l+2}-\lambda_l)(d_{l+2}-d_l)=-8(d_l+d_{l+2})-a_5$, an identity for \eqref{eq:dl}. On the diagonal it is the first-order recurrence
\begin{equation}\label{eq:rec}
8(s-3)\,x(s-2)-8(s+1)\,x(s)-8d(s)^2-a_5d(s)=128\lambda^3-a_3\lambda^2-a_6\lambda-a_9,\qquad \lambda=s(s-2)+\tfrac34,
\end{equation}
where $x(s)=x_l$, $d(s)=d_l$ with $s=l+\tfrac32$, and \eqref{eq:5pairs} reads $x(s)=(H^2-s^2)\prod_{c}(s^2-c^2)/s^2(s^2-1)$ over $c\in\{M_\pm\pm\tfrac12\}$. The left side of \eqref{eq:rec} is a rational function of $s$ whose denominator $s^2(s-2)^2$ cancels, and it is a polynomial of degree $3$ in $\lambda$ with leading coefficient $128$, while its three remaining coefficients are the values $a_3$, $a_6$, $a_9$ of Theorem~\ref{thm:HR}, in agreement with the reduction of Proposition~\ref{prop:status}.

\subsection{Representations}\label{sec:proofreps}
For Theorem~\ref{thm:reps}, the transversality \eqref{eq:transvp} and Lemma~\ref{lem:red} give $\alpha_l=\alpha'_l$, i.e.\ $R_{ll}$ as stated, for every $l\ge1$. For $l\in\{0,\tfrac12\}$ the block $(l,l)$ of a rank-2 tensor vanishes. The diagonal blocks of \eqref{eq:TT3p} and \eqref{eq:TT2p} are \eqref{eq:id3} and \eqref{eq:id2}, linear in $\Pc_l$ and $\Pc_{l-2}$ with the coefficients displayed in \eqref{eq:red3} and \eqref{eq:red2}. Their determinant,
\begin{align*}
&-\frac{1}{(l+1)(l+2)(2l+1)(2l+3)}\,\frac{2}{l(2l-1)(2l+1)}\\
&\quad-\frac{1}{l(l-1)(2l-1)(2l+1)}\,\frac{2}{(l+1)(2l+1)(2l+3)}\\
&=-\frac{2}{l(l-1)(l+1)(l+2)(2l-1)(2l+1)(2l+3)}\,,
\end{align*}
does not vanish for $l\ge2$. Hence, going down a chain from its top $l_1$, where $\Pc_{l_1}=0$ since there is no state above, the two relations determine $\Pc_{l_1-2}$, then $\Pc_{l_1-4}$, and so on: $\Pc_l=\Pc(l)$ for every $l$ in the chain, since \eqref{eq:Pcl} is the unique solution, and $\Pc(l_1)=0$ is required. By \eqref{eq:Pcl}, $\Pc(l_1)=0$ if and only if $l_1\in\{k,\,|q|-1,\,|q|-2\}$ (the factors $k-l$, $(l+1)^2-q^2$, $(l+2)^2-q^2$). At the bottom $l_0$ of the chain the relations on the block $(l_0,l_0)$ are the identities \eqref{eq:id3}--\eqref{eq:id0} without the terms in $\Pc(l_0-2)$, since there is no state below. As the identities hold with these terms, the relations hold if and only if the terms vanish. For $l_0\ge\tfrac32$ the coefficients of $\Pc(l_0-2)$ in \eqref{eq:red0}--\eqref{eq:red3} are finite and nonzero, and
\[
\Pc(l_0-2)=\frac{16\,l_0(l_0-1)\,(l_0^2-q^2)\,((l_0-1)^2-q^2)\,(k-l_0+2)(k+l_0+1)}{2l_0-1}
\]
vanishes if and only if $l_0\in\{|q|,\,|q|+1,\,k+2\}$ or $l_0\in\{0,1\}$, so in particular $l_0=\tfrac32$ requires $|q|\in\{\tfrac12,\tfrac32\}$, i.e.\ a physical chain. For $l_0\in\{0,1\}$ some blocks are absent, those with $L\ge1$ at $l_0=0$ and the rank-3 block at $l_0=1$ ($j_3(1)=0$), and the remaining relations are checked directly below. For $l_0=\tfrac12$, finally, $T'$ has no block $(\tfrac12,\tfrac12)$, so that $\J\cdot T'$ vanishes there while $\tfrac13(4q^2-1)H\J$ does not unless $q^2=\tfrac14$: the transversality excludes $l_0=\tfrac12$ except for the physical chain of $|q|=\tfrac12$. For completeness, the scalar relation on the singlet $\J=T'=0$ reads $0=\tfrac83(q^2-1)(q^2-4)(4H^2-9)$. Thus $|q|=1,2$ admit the small singlet, while $q=0$ has the physical singlet at $H=\tfrac32$. On a single top-spin $l=1$ block the residual scalar and rank-1/rank-2 equations contain the factor $(4H^2-25)(q^2-4)(q^2-9)$. The top-spin $l=\tfrac12$ check yields $q^2=\tfrac14$, $H=2$, and the rank equations likewise cover top spin $l=\tfrac32$ without a singular division. These direct checks complete the argument at the lowest angular momenta.\par

The remaining relations, on the blocks $(l+2,l)$ and the diagonal blocks with $L=0,1$, are then identities for the admissible chains. Collecting: the finite chains $\{l_0,l_0+2,\dots,l_1\}$ carrying a representation are those with $l_0\in\{|q|,|q|+1\}$ or $l_0\in\{0,1\}$ and $l_1\in\{k,|q|-1,|q|-2\}$, subject to $l_0\le l_1$ and $l_1-l_0$ even, and the only infinite chain is $\{k+2,k+4,\dots\}$, whose bottom $l_0=k+2$ is the remaining zero of $\Pc(l_0-2)$ and which has no top. With the condition $P_l>0$ along the chain for the irreducible hermitian ones, this gives the physical levels for every $q$, the chains $\{0\text{ or }1,\dots,|q|-2\text{ or }|q|-1\}$ for integer $|q|\ge2$, the singlet $\{0\}$ at $|q|=1$, none of the second kind for half-integer $q$, and the non-unitary chain above $k$, where $\Pc(l)<0$. As in Theorem~\ref{thm:reps}, the classification is made at the energy of the level $k$ and for angular momenta of its parity.

\section{Concluding remarks}\label{sec:conclusion}

The charged oscillator in the field of a monopole has, as symmetry algebra, the $\so(3)$-covariant extension of a Heun--Racah algebra. The algebra is generated by the angular momentum and by the rank-2 tensor $T'$, its relations are polynomial in $\J$ with coefficients in $H$ and $q^2$, and it is not quadratic. In each level and each sector, $T'_{zz}$ is an algebraic Heun operator of a Hahn pair whose first member is $\J^2$, the pair $(\J^2,T'_{zz})$ generates the cubic algebra $\mathcal A_{\rm gen}$ with $a_0=0$, and at $q=0$ the tensor is the algebraic Heun operator of the pair (Casimir of $\so(3)$, quadrupole of $\su(3)$). The monopole charge enters only through the parameters, and symmetrically with the magnetic quantum number after the additive scalar shift $T'_{zz}\mapsto T'_{zz}+\tfrac23Hq^2\mathbf1$, a symmetry realised locally on the wave functions. The representations of the algebra reproduce the levels of the system and reveal, for integer charges, the representations carried by the states that the monopole removes, which are the kernel of a power of the Newman--Penrose operator. In the sectors, the integrals obtained by the recurrence approach in \cite{HMZ2018} and their deformed-oscillator algebra are recovered, up to normalisation, as the parts of $T'_{zz}$ that raise and lower the angular momentum, and the algebra is identified. The cubic Casimir of the sector algebra is also obtained. Everything rests on the dynamical algebra of the ladder operators, which is established by operator identities, the mechanism being that the centrifugal term of McIntosh and Cisneros makes the anomalous commutators proportional to $\J$. The action of the ladder operators and the covariant relations of the symmetry algebra follow from that algebra alone, the latter by reordering, and the Heun--Racah relations of the sectors follow from the covariant relations.

The questions left open are of two kinds. On the algebraic side, the presentation of $\Alg(H,q)$ is known to be compatible with the representations exhibited here but not to be complete, and the existence of a Poincar\'e--Birkhoff--Witt basis, as well as the representations with multiplicities or of mixed parity, has not been examined. The derivation of Section~\ref{sec:reorder} shows that the relations of Theorem~\ref{thm:T} are consequences of Proposition~\ref{prop:dyn} alone, and one may ask whether they generate all the relations of the symmetry algebra that follow from it. On the side of the model, the classical limit, where the Poisson brackets of the corresponding functions should give a Poisson analogue of the Heun--Racah structure, and the Cartesian form of the rank-3 relation \eqref{eq:TT3} remain to be worked out. The occurrence, in a given sector, of several Hahn pairs of which $T'_{zz}$ is simultaneously a Heun operator, related by rational gauges, deserves to be understood. The system admits a finite version in which the radial variable lives on a quadratic lattice, with the same symmetry algebra, which will be presented separately.

\section*{Data availability}
No data were generated. The computer algebra scripts with which the identities of Section~\ref{sec:proofs} were checked are available from the corresponding author upon request.

\section*{Conflict of interest}
The authors have no conflicts of interest to disclose.

\section*{Acknowledgements}
The research of L.V. is supported in part by a Discovery Grant from the Natural Sciences and Engineering Research Council (NSERC) of Canada. IM was supported by the Australian Research Council Future Fellowship FT180100099. It is a pleasure for L.V. to thank Simon Labelle and Michel Mayrand, co-authors of \cite{LMV}, where the question settled here was first posed. An AI assistant (Claude, Anthropic) was used in the derivation and verification of the identities of Section~\ref{sec:proofs} and in the drafting of the text; the authors have verified all results and take full responsibility for the content.

\appendix
\section{Conventions and closed forms}\label{app:6j}
Clebsch--Gordan coefficients are those of Condon and Shortley. For a vector operator, $\langle l\,m\,1\,0|l+1\,m\rangle=\sqrt{\dfrac{(l+1-m)(l+1+m)}{(2l+1)(l+1)}}$ and $\langle l\,m\,1\,0|l-1\,m\rangle=-\sqrt{\dfrac{(l-m)(l+m)}{l(2l+1)}}$. For a rank-2 tensor,
\begin{gather*}
\langle l\,m\,2\,0|l+2\,m\rangle^2=\frac{3(l+1-m)(l+2-m)(l+1+m)(l+2+m)}{2(2l+1)(l+1)(2l+3)(l+2)},\\
\langle l+2\,m\,2\,0|l\,m\rangle=\sqrt{\frac{2l+1}{2l+5}}\,\langle l\,m\,2\,0|l+2\,m\rangle ,\qquad
\langle l\,m\,2\,0|l\,m\rangle=\frac{3m^2-l(l+1)}{\sqrt{(2l-1)\,l(l+1)\,(2l+3)}} .
\end{gather*}
The reduced matrix elements of $J^{(L)}$ in the convention $\langle l'm'|X^{(L)}_M|lm\rangle=\langle l\,m\,L\,M|l'm'\rangle X_{l'l}$ are
\begin{gather*}
j_1(l)=\sqrt{l(l+1)},\qquad j_2(l)=\sqrt{\frac{(2l-1)\,l(l+1)\,(2l+3)}{6}},\\
j_3(l)=(-1)^{2l+3}\sqrt{7(2l+1)}\,\sixj{l}{2}{l}{1}{l}{3}\,j_1(l)\,j_2(l)=\sqrt{\frac{(l-1)\,l(l+1)\,(l+2)\,(2l-1)(2l+3)}{10}} .
\end{gather*}
The $6j$ symbols used in Section~\ref{sec:proofs} all have an entry equal to $1$ or $2$ and are given by the general formulas of \cite[Table 5]{Edmonds}. For integer $l$ they read, with $\lambda=l(l+1)$,
\begin{small}
\begin{gather*}
\sixj{l}{1}{l+1}{1}{l+2}{2}=\frac{1}{\sqrt{5(2l+3)}},\quad \sixj{l}{1}{l+1}{1}{l}{2}=\sqrt{\frac{l(2l-1)}{30(l+1)(2l+1)(2l+3)}},\\
\sixj{l}{1}{l-1}{1}{l}{2}=\sqrt{\frac{(l+1)(2l+3)}{30\,l(2l-1)(2l+1)}},\quad \sixj{l}{1}{l+1}{1}{l}{1}=\sqrt{\frac{l}{6(l+1)(2l+1)}},\\
\sixj{l}{1}{l-1}{1}{l}{1}=-\sqrt{\frac{l+1}{6\,l(2l+1)}},\quad \sixj{l}{2}{l}{1}{l}{1}=\sqrt{\frac{(2l-1)(2l+3)}{30\,\lambda(2l+1)}},\\
\sixj{l}{2}{l}{1}{l}{3}=-\sqrt{\frac{3(\lambda-2)}{35\,\lambda(2l+1)}},\quad \sixj{l}{2}{l+2}{1}{l+2}{3}=-\sqrt{\frac{l}{21(l+2)(2l+5)}},\\
\sixj{l}{2}{l''}{2}{l}{0}=\frac{1}{\sqrt{5(2l+1)}}\ \ (l''=l,l\pm2),\quad \sixj{l}{2}{l+2}{2}{l}{1}=\sqrt{\frac{2l}{15(l+1)(2l+1)}},\\
\sixj{l}{2}{l}{2}{l}{1}=-\sqrt{\frac{3}{10\,\lambda(2l+1)}},\quad \sixj{l}{2}{l-2}{2}{l}{1}=-\sqrt{\frac{2(l+1)}{15\,l(2l+1)}},\\
\sixj{l}{2}{l+2}{2}{l}{2}=\sqrt{\frac{2l(2l-1)}{35(l+1)(2l+1)(2l+3)}},\quad \sixj{l}{2}{l}{2}{l}{2}=-\frac{(2l-3)(2l+5)}{\sqrt{70\,\lambda(2l-1)(2l+1)(2l+3)}},\\
\sixj{l}{2}{l-2}{2}{l}{2}=\sqrt{\frac{2(l+1)(2l+3)}{35\,l(2l-1)(2l+1)}},\quad \sixj{l}{2}{l+2}{2}{l}{3}=\sqrt{\frac{(l-1)l(2l-1)}{70(l+1)(l+2)(2l+1)(2l+3)}},\\
\sixj{l}{2}{l}{2}{l}{3}=6\sqrt{\frac{2(\lambda-2)}{35\,\lambda(2l-1)(2l+1)(2l+3)}},\\
\sixj{l}{2}{l-2}{2}{l}{3}=-\sqrt{\frac{(l+1)(l+2)(2l+3)}{70\,l(l-1)(2l-1)(2l+1)}},\quad
\sixj{l}{2}{l+2}{2}{l+2}{3}=\sqrt{\frac{l(2l+7)}{14(l+2)(2l+3)(2l+5)}},\\
\sixj{l}{2}{l+2}{2}{l+2}{2}=\sqrt{\frac{2(l+3)(2l+7)}{35(l+2)(2l+3)(2l+5)}},\qquad \sixj{l}{2}{l}{2}{l+2}{3}=-\sqrt{\frac{(l+3)(2l-1)}{14(l+1)(2l+1)(2l+3)}}\,,
\end{gather*}
\end{small}
and $\sixj{l}{2}{l}{2}{l+2}{2}=\sixj{l}{2}{l+2}{2}{l}{2}$. For half-integer $l$ each of these symbols is multiplied by $-1$. All closed forms are used only for admissible triangles. A forbidden tensor block is set to zero before evaluating a formula with a vanishing denominator.
The components $b^\dagger_{\pm1}=\mp(b^\dagger_x\pm ib^\dagger_y)/\sqrt2$ act by $\langle n',l',m\pm1|b^\dagger_{\pm1}|n,l,m\rangle=\langle l\,m\,1\,\pm1|l'\,m\pm1\rangle\,\rho_{(n',l'),(n,l)}$ with $\rho$ as in \eqref{eq:rho}, which reproduces eqs.~(4.10)--(4.11) of \cite{LMV}.

\end{document}